\documentclass{amsart}

\usepackage[a4paper]{geometry}
\usepackage{amsfonts}
\usepackage{mathrsfs}
\usepackage{amsmath}
\usepackage{amsthm}
\usepackage{amssymb}
\usepackage{latexsym}
\usepackage[all]{xy}
\usepackage{graphicx}

\usepackage[notref,notcite]{showkeys}

\numberwithin{equation}{section}

\newtheorem{thrm}{Theorem}[section]
\newtheorem{prop}{Proposition}[section]
\newtheorem{defn}{Definition}[section]

\theoremstyle{definition}
\newtheorem{rmk}{Remark}[section]

\newcommand{\nc}{\newcommand}

\nc{\al}{\alpha}
\nc{\be}{\beta}
\nc{\eps}{\epsilon}
\nc{\veps}{\varepsilon}
\nc{\ga}{\gamma}
\nc{\Ga}{\Gamma}
\nc{\ka}{\kappa}
\nc{\la}{\lambda}
\nc{\La}{\Lambda}
\nc{\del}{\delta}
\nc{\om}{\omega}
\nc{\si}{\sigma}
\nc{\Ups}{\upsilon}
\nc{\vphi}{\varphi}

\nc{\id}{\mathrm{id}}
\nc{\gr}{\mathrm{gr}}
\nc{\Aut}{\mathrm{Aut}}

\nc{\tl}{\tilde}

\nc{\mbA}{\mathbf{A}}
\nc{\mbb}{\mathbf{b}}
\nc{\mbB}{\mathbf{B}}
\nc{\mbc}{\mathbf{c}}
\nc{\mbC}{\mathbf{C}}
\nc{\mbd}{\mathbf{d}}
\nc{\mbD}{\mathbf{D}}
\nc{\mbe}{\mathbf{e}}
\nc{\mbE}{\mathbf{E}}
\nc{\mbf}{\mathbf{f}}
\nc{\mbF}{\mathbf{F}}
\nc{\mbg}{\mathbf{g}}
\nc{\mbH}{\mathbf{H}}
\nc{\mbh}{\mathbf{h}}
\nc{\mbi}{\mathbf{i}}
\nc{\mbI}{\mathbf{I}}
\nc{\mbj}{\mathbf{j}}
\nc{\mbJ}{\mathbf{J}}
\nc{\mbk}{\mathbf{k}}
\nc{\mbK}{\mathbf{K}}
\nc{\mbL}{\mathbf{L}}
\nc{\mbM}{\mathbf{M}}
\nc{\mbQ}{\mathbf{Q}}
\nc{\mbq}{\mathbf{q}}
\nc{\mbr}{\mathbf{r}}
\nc{\mbT}{\mathbf{T}}
\nc{\mbu}{\mathbf{u}}
\nc{\mbU}{\mathbf{U}}
\nc{\mbv}{\mathbf{v}}
\nc{\mbV}{\mathbf{V}}
\nc{\mbw}{\mathbf{w}}
\nc{\mbW}{\mathbf{W}}
\nc{\mbX}{\mathbf{X}}
\nc{\mbY}{\mathbf{Y}}
\nc{\mbZ}{\mathbf{Z}}

\nc{\mbbA}{\mathbb{A}}
\nc{\mbbB}{\mathbb{B}}
\nc{\mbbD}{\mathbb{D}}
\nc{\mbbF}{\mathbb{F}}
\nc{\mbbV}{\mathbb{V}}
\nc{\mbbH}{\mathbb{H}}
\nc{\mbbK}{\mathbb{K}}
\nc{\mbbL}{\mathbb{L}}
\nc{\mbbP}{\mathbb{P}}
\nc{\mbbU}{\mathbb{U}}

\nc{\mcA}{\mathcal{A}}
\nc{\mcB}{\mathcal{B}}
\nc{\mcC}{\mathcal{C}}
\nc{\mcD}{\mathcal{D}}
\nc{\mcE}{\mathcal{E}}
\nc{\mcF}{\mathcal{F}}
\nc{\mcH}{\mathcal{H}}
\nc{\mcK}{\mathcal{K}}
\nc{\mcO}{\mathcal{O}}
\nc{\mcQ}{\mathcal{Q}}
\nc{\mcS}{\mathcal{S}}
\nc{\mcP}{\mathcal{P}}
\nc{\mcU}{\mathcal{U}}
\nc{\mcT}{\mathcal{T}}
\nc{\mcV}{\mathcal{V}}
\nc{\mcY}{\mathcal{Y}}
\nc{\mcZ}{\mathcal{Z}}

\nc{\mfa}{\mathfrak{a}}
\nc{\mfA}{\mathfrak{A}}
\nc{\mfb}{\mathfrak{b}}
\nc{\mfB}{\mathfrak{B}}
\nc{\mfC}{\mathfrak{C}}
\nc{\mfd}{\mathfrak{d}}
\nc{\mfD}{\mathfrak{D}}
\nc{\mfe}{\mathfrak{e}}
\nc{\mfE}{\mathfrak{E}}
\nc{\mff}{\mathfrak{f}}
\nc{\mfF}{\mathfrak{F}}
\nc{\mfg}{\mathfrak{g}}
\nc{\mfgl}{\mathfrak{g}\mathfrak{l}}
\nc{\mfh}{\mathfrak{h}}
\nc{\mfH}{\mathfrak{H}}
\nc{\mfJ}{\mathfrak{J}}
\nc{\mfk}{\mathfrak{k}}
\nc{\mfK}{\mathfrak{K}}
\nc{\mfl}{\mathfrak{l}}
\nc{\mfL}{\mathfrak{L}}
\nc{\mfM}{\mathfrak{M}}
\nc{\mfm}{\mathfrak{m}}
\nc{\mfn}{\mathfrak{n}}
\nc{\mfN}{\mathfrak{N}}
\nc{\mfo}{\mathfrak{o}}
\nc{\mfP}{\mathfrak{P}}
\nc{\mfQ}{\mathfrak{Q}}
\nc{\mfS}{\mathfrak{S}}
\nc{\mfsl}{\mathfrak{s}\mathfrak{l}}
\nc{\mfso}{\mathfrak{s}\mathfrak{o}}
\nc{\mfsp}{\mathfrak{s}\mathfrak{p}}
\nc{\mft}{\mathfrak{t}}
\nc{\mfU}{\mathfrak{U}}
\nc{\mfu}{\mathfrak{u}}
\nc{\mfV}{\mathfrak{V}}
\nc{\mfX}{\mathfrak{X}}
\nc{\mfY}{\mathfrak{Y}}
\nc{\mfz}{\mathfrak{z}}

\nc{\msA}{\mathsf{A}}
\nc{\msB}{\mathsf{B}}
\nc{\msC}{\mathsf{C}}
\nc{\msc}{\mathsf{c}}
\nc{\msD}{\mathsf{D}}
\nc{\msd}{\mathsf{d}}
\nc{\mse}{\mathsf{e}}
\nc{\msw}{\mathsf{w}}
\nc{\msq}{\mathsf{q}}
\nc{\msg}{\mathsf{g}}
\nc{\msE}{\mathsf{E}}
\nc{\msf}{\mathsf{f}}
\nc{\msF}{\mathsf{F}}
\nc{\msh}{\mathsf{h}}
\nc{\msH}{\mathsf{H}}
\nc{\msI}{\mathsf{I}}
\nc{\msJ}{\mathsf{J}}
\nc{\msK}{\mathsf{K}}
\nc{\msL}{\mathsf{L}}
\nc{\msP}{\mathsf{P}}
\nc{\msQ}{\mathsf{Q}}
\nc{\msR}{\mathsf{R}}
\nc{\mss}{\mathsf{s}}
\nc{\msS}{\mathsf{S}}
\nc{\msT}{\mathsf{T}}
\nc{\msU}{\mathsf{U}}
\nc{\msV}{\mathsf{V}}
\nc{\msX}{\mathsf{X}}
\nc{\msY}{\mathsf{Y}}
\nc{\msZ}{\mathsf{Z}}

\nc{\End}{\mathrm{End}}
\nc{\Ext}{\mathrm{Ext}}
\nc{\Hom}{\mathrm{Hom}}
\nc{\Ima}{\mathrm{Image}}
\nc{\Ind}{\mathrm{Ind}}
\nc{\Ker}{\mathrm{Ker}}
\nc{\RHom}{\mathrm{RHom}}
\nc{\Sym}{\mathrm{Sym}}

\nc{\mf}{\mathfrak}
\nc{\mc}{\mathcal}
\nc{\ms}{\mathsf}
\nc{\bb}{\mathbb}

\nc{\wh}{\widehat}
\nc{\wt}{\widetilde}

\nc{\Q}{\mathbb{Q}}
\nc{\C}{\mathbb{C}}
\nc{\N}{\mathbb{N}}
\nc{\Z}{\mathbb{Z}}

\nc{\ot}{\otimes}
\nc{\op}{\oplus}
\nc{\ol}{\overline}

\nc{\lan}{\langle}
\nc{\ran}{\rangle}

\nc{\spl}[1]{\begin{equation}\begin{aligned}#1\end{aligned}\end{equation}}
\nc{\eqa}[1]{\begin{align}#1\end{align}}
\nc{\eq}[1]{\begin{equation}#1\end{equation}}
\nc{\eqn}[1]{\begin{align*}#1\end{align*}}

\nc{\tx}[1]{\qu\text{#1}\qu}
\nc{\sqb}[1]{\text{\tiny$($}{#1}\text{\tiny$)$}} 

\nc\el{\nonumber\\}
\nc\nn{\nonumber}

\nc{\sm}[1]{\text{\tiny{\rm #1}}}

\nc{\qu}{\quad}
\nc{\qq}{\qquad}

\usepackage{accents}
\nc*{\dt}[1]{%
  \accentset{\mbox{\large\bfseries .}}{#1}}

\usepackage{color}

\nc{\red}{\color{red}}
\nc{\blu}{\color{blue}}

\begin{document}

\begin{flushright}
DMUS-MP-15/05
\end{flushright}

\vspace{.5cm}

\title{Yangian of $AdS_3 /CFT_2$ and its deformation}

\author{Vidas Regelskis}

\address{Department of Mathematics, University of Surrey, Guildford, GU2 7YX, U.K.}

\email{v.regelskis@surey.ac.uk}

\begin{abstract}

We construct highest-weight modules and a Yangian extension of the centrally extended $\mf{sl}(1|1)^2$ superalgebra, that is a symmetry of the worldsheet scattering associated with the $AdS_3 /CFT_2$ duality. We demonstrate that the R-matrix intertwining atypical modules has an elegant trigonometric parametrization. We also consider a quantum deformation of this superalgebra, its modules, and obtain a quantum affine extension of the Drinfeld-Jimbo type that describes a deformed worldsheet scattering. 

\end{abstract}


\maketitle

\thispagestyle{empty}


\section{Introduction}

Recent progress in exploring integrability in $AdS/CFT$ dualities has led to the discovery of many new algebraic structures. One of the most notable is a class of the so-called $u$-deformed Hopf superalgebras, which emerge in the worldsheet scattering theory in various backgrounds \cite{BSS1,BSSST1,GH,HPT,PST,Sf}. These superalgebras are deformed in the direction of their central extensions and lead to $R$-matrices of a non-relativistic type closely resembling that of the one-dimensional Hubbard chain \cite{Be2,MaMe}. 

For example, the worldsheet superalgebra of the $AdS_5 /CFT_4$ duality is the centrally extended superalgebra $\mfsl({2|2})\op{\C u^{\pm}}$ admitting a $u$-deformed Hopf algebra structure \cite{PST} and having a non-conventional representation theory \cite{MaMo}. Interestingly, it admits a non-standard $u$-deformed Yangian extension, which was constructed in various realizations: the Drinfeld $J$ presentation \cite{Be3}, Drinfeld New presentation \cite{ST}, and $RTT$-presentation \cite{BL}. Moreover, this superalgebra can be further deformed in the Cartan direction. In such a way one obtains a double-deformed Hopf superalgebra. In particular, a $q$-deformation of the $u$-deformed $\mfsl({2|2})\op{\C u^{\pm}}$ and its affinization of the Drinfeld-Jimbo type were constructed in \cite{BK} and \cite{BGM}, respectively.  

In this paper we consider the centrally extended superalgebra $\C h_0\ltimes\mfsl({1|1})^2\op{\C u^{\pm}}$, which was shown to be a symmetry of the worldsheet scattering in the $AdS_3 /CFT_2$ duality \cite{BSZ} for the $AdS_3 \times S^3 \times S^3 \times S^1$ background \cite{BSS1,BSSS2}. Moreover, its quantum deformation was shown to be a symmetry of the deformed worldsheet scattering \cite{Ho}. This superalgebra also serves as a prototype for the symmetries of the duality on the $AdS_3 \times S^3 \times T^4$ background \cite{BSSS1,BSSST1}, which essentially contains two copies of this superalgebra with their central elements identified. For this reason, in this paper we will focus on the algebraic constructions associated with the former case only. 

This paper contains two parts. In the first part we construct highest-weight modules and a \mbox{$u$-deformed} Yangian extension of the extended superalgebra $\C h_0\ltimes\mfsl({1|1})^2\op{\C u^{\pm}}$. In particular, we construct typical and atypical Kac modules $K(\la_1,\la_2,\nu)$ and $A(\la_1,\la_2,\nu)$, and the one-dimensional module $\ms{1}$ of $\mfsl({1|1})^2\op{\C u^{\pm}}$. We show that the tensor product of two atypical modules is isomorphic to the typical one and obtain the corresponding $R$-matrix. In a suitable parametrization, this $R$-matrix has a trigonometric form described by three independent variables. We also construct an evaluation homomorphism from the newly constructed Yangian to the universal enveloping algebra of $\C h_0\ltimes\mfsl({1|1})^2\op{\C u^{\pm}}$.

In the second part of the paper we consider a quantum deformation of the $\C h_0\ltimes\mfsl({1|1})^2\op{\C u^{\pm}}$ superalgebra and obtain its affine extension of the Drinfeld-Jimbo type. The structure of the second part closely resembles that of the first part. We construct quantum deformed analogues of the highest-weight modules constructed before and obtain the deformed $R$-matrix. The affinization presented in this paper is inspired by a similar double-deformed construction presented in \cite{BGM}.

The main results of this paper are presented in Sections 3 and 5, where the Yangian extension of the superalgebra $\C h_0\ltimes\mfsl({1|1})^2\op{\C u^{\pm}}$ and an affinization of its quantum deformation are presented. Sections 2 and 5 serve as the necessary preliminaries. Appendices A and D contain some additional computations and formulae that were omitted in the main parts of the manuscript. Appendices B and C explain the connection between the notation used in the present paper and the traditional notation which uses the $x^\pm$ variables. 

The goal of the present study was to obtain new infinite dimensional superalgebras and deformed superalgebras that can be further used to study the highest-weight representation theory along the lines of \cite{HZh,RZh1,RZh2}. Such representations would be important in studying integrability of the $AdS_3/CFT_2$ duality using the techniques of the Bethe Ansatz similar to the ones introduced in \cite{BR,MaMe,RS}, and progress towards the Baxter Q-operators using the methods introduced in \cite{BFLMS,FH,HJ}. 

There is also a number of other important aspects of the integrability in the $AdS_3/CFT_2$ duality, where extended symmetries could play an important role: the description of the massless modes \cite{BSSS1,BSSS2,SST}, determination of the so-called dressing phases \cite{BSSST2}, integrability in the presence of mixed-flux backgrounds \cite{BSSS2,LSSS,HST,HT} and deformations \cite{HRT,Ho}. The latter questions go beyond the scope of the present paper and will not be considered. We leave these question for further study.


\section{The superalgebra $\C \ltimes \mfsl({1|1})^2\op{\C u^{\pm}}$ and its highest-weight modules} \label{Sec:2}

In this section we present the superalgebra $\mfa=\C \ltimes \mfsl({1|1})^2\op{\C u^{\pm}}$ and the associated Hopf algebra which arises as a symmetry of the worldsheet scattering in the $AdS_3 /CFT_2$ duality \cite{BSS1}. We then construct highest-weight modules of this algebra that are important in the aforementioned duality. 

\subsection{Algebra}

Let $[\cdot\hspace{.5mm},\cdot]$ denote the $\Z_2$-graded commutator, i.e.\ $[a,b]=ab-(-1)^{p(a)p(b)}ba$ for $\forall a,b\in\mfg$, where $\mfg$ is a Lie superalgebra and $p=\deg_2$ denotes the $\Z_2$-grading on $\mfg$. We will also use the notation $\C^\times = \C \backslash \{0\}$. 

\medskip

We start by considering the centrally extended superalgebra $\C \ltimes \mfsl({1|1})^2\op \C^2$, where $\C^2 = \C k_1 \op \C k_2$. We then obtain $\C \ltimes \mfsl({1|1})^2\op{\C u^{\pm}}$ as the quotient of an extension of the former algebra. The motivation for this approach is explained in Remark \ref{R:2.2} given below.

\begin{defn}
The centrally extended superalgebra $\C \ltimes \mfsl({1|1})^2\op\C^2$ is generated by elements $e_i$, $f_i$, $h_0$ and central elements $h_i$, $k_i$ with $i,j\in\{1,2\}$ satisfying 
\eq{
[ e_i , f_j ] = \del_{ij} h_i + (1-\del_{ij})\, k_i, \qu [ h_0, f_i ] = - f_i, \qu [ h_0, e_i ] = e_i .   \label{a:Lie}
}
The remaining relations are trivial. The $\Z_2$-grading is given by $\deg_2(h_0)=\deg_2(h_i)=\deg_2(k_i)=0$ and $\deg_2(e_i)=\deg_2(f_i)=1$.
\end{defn}

\begin{rmk} \label{R:aut-a}
The algebra $\C \ltimes \mfsl({1|1})^2\op\C^2$ has outer-automorphism group $GL(2)^2$ acting by
\eq{
\left(\!\!\begin{array}{c} e_1 \\ e_2 \end{array}\!\!\right) \mapsto A \left(\!\!\begin{array}{c} e_1 \\ e_2 \end{array}\!\!\right) , \qu 
\left(\!\!\begin{array}{c} f_1 \\ f_2 \end{array}\!\!\right) \mapsto B \left(\!\!\begin{array}{c} f_1 \\ f_2 \end{array}\!\!\right) , \qu
\left(\!\!\begin{array}{cc} h_1 & k_1 \\ k_2 & h_2 \end{array}\!\!\right) \mapsto A \left(\!\!\begin{array}{cc} h_1 & k_1 \\ k_2 & h_2 \end{array}\!\!\right) B^t , \label{aut-a}
}
for any $(A,B)\in GL(2)^2$. Here $B^t$ denotes the transposed matrix. The element $h_0$, which acts as an outer-automorphism on the subalgebra $\mfsl({1|1})^2\op\C^2$, is invariant under the action of $GL(2)^2$.
\end{rmk}

Our focus will be on the tensor product of two atypical highest-weight modules and the $R$-matrix. Bearing in mind this goal we extend the algebra above by the ring $\C u^{\pm} (=\C u^+ \op \C u^-)$ such that $u^\pm u^\mp=1$ and introduce a book-keeping notation $\mfa_0 = \C \ltimes \mfsl({1|1})^2\op\C^2\op\C u^{\pm}$. Note that $u^\pm$ are central in $\mfa_0$. 
Let $U(\mfa_0)$ denote the universal enveloping algebra of $\mfa_0$. The next observation is immediate from the defining relations of the algebra.

\begin{prop}
The vector space basis of $U(\mfa_0)$ is given in terms of monomials
\spl{ \label{PBW}
& (f_2)^{r_2} (f_1)^{r_1} (h_0)^{l_0}(h_1)^{l_1} (h_2)^{l_2} (k_1)^{l_3} (k_2)^{l_4} (u)^{t} (e_1)^{s_1} (e_2)^{s_2} 
}
with $r_i,s_i \in \{0,1\}$, $l_i\in \Z_{\ge0}$ and $t\in\Z$.
\end{prop}

The monomials \eqref{PBW} give a Poincar\'e--Birkhoff--Witt type basis of $U(\mfa_0)$. Moreover, by the standard arguments, $U(\mfa_0)\cong U^-_0.U^0_0 .U^+_0$ as vector spaces, where $U^-_0$ and $U^+_0$  are the nilpotent subalgebras generated by elements $f_i$ and $e_i$ with $i\in\{1,2\}$, respectively, and $U^0_0$ is generated by the remaining elements of $\mfa_0$.

\begin{rmk}  \label{R:Z5}
The algebra $U(\mfa_0)$ also admits a $\Z$-grading given by
$$
\deg(h_0) = \deg(h_i) = \deg(u^{\pm}) = 0, \qu \deg(e_i)= \dt\pm1, \qu \deg(f_{i})=\dt\mp1, \qu \deg(k_i)= \dt\pm 2,
$$
where the upper sign in $\dt\pm$ and $\dt\mp$ is for $i=1$ and the lower sign is for $i=2$ (we will use this dotted notation throughout this paper). Note that we could equivalently define the $\Z$-grading by inverting the grading, namely $\deg \to -\deg$. 
\end{rmk}

Let $I_0$ be the ideal of $U(\mfa_0)$ generated by the relation
\eq{ \label{I0}
k_i = \al_i (u^2-u^{-2}),
}
where $\al_i\in\C^\times$. Define the quotient algebra $U(\mfa) = U(\mfa_0)/I_0$. Then one can introduce a Hopf algebra structure on $U(\mfa)$ given by the coproduct for $i\in\{1,2\}$ (and $i=0$ for $h_i$)
\spl{ \label{a:cop}
\Delta(e_i) &= e_i \ot u^{\dt\mp} + u^{\dt\pm} \ot e_i , \qq  \Delta(k_i) = k_i \ot u^{\dt\mp2} + u^{\dt\pm2} \ot k_i , \\
\Delta(f_i) &= f_i \ot u^{\dt\pm} + u^{\dt\mp} \ot f_i , \qq \Delta(h_i) = h_i \ot 1 + 1 \ot h_i ,\qq \Delta(u^\pm)=u^\pm\ot u^\pm ,
}
the counit $\veps(a)=0$ and the antipode $S(a) = - a$ for all $a\in\mfa$ except $\veps(u^\pm)=1$, $S(u^\pm)=u^\mp$. 

\medskip

Let us pinpoint some properties of the Hopf algebra $U(\mfa)$. First, the $\Z$-grading introduced in Remark \ref{R:Z5} agrees with the powers of $u^\pm$ in the coproduct \eqref{a:cop}. Second, we want to explain the role of the ideal~$I_0$. Set $\Delta^{\rm op} = \si \circ \Delta$, where $\si : a \ot b \mapsto (-1)^{p(a)p(b)}b \ot a$ is the graded permutation operator. Observe that all central elements of $U(\mfa)$ are co-commutative:
$$
\Delta(c) = \Delta^{\rm op}(c) \tx{for} c\in\{h_i,k_i,u^\pm\} , \qu i\in\{1,2\}. 
$$
This is obvious for $h_i$ and $u^\pm$, while for $k_i$ we have $\Delta(k_i)-\Delta^{\rm op}(k_i)=k_i \ot (u^{\dt\mp2} - u^{\dt\pm2}) + (u^{\dt\pm2} - u^{\dt\mp2}) \ot k_i$, which is equal to zero only if \eqref{I0} holds. (Moreover, this is a necessary condition for existence of the $R$-matrix.) Also note that due to \eqref{I0} the algebra $U(\mfa)$ is actually a two-parameter family of Hopf algebras parametrized by $\al_i$. Moreover, we need to formally set $\deg\al_i=\dt\pm2$, for \eqref{I0} to respect the $\Z$-grading. (In the $AdS_3/CFT_2$ duality one usually sets $\al_1=-\al_2=-h/2$, where $h\in\C^\times$ plays the role of the coupling constant of the underlying field theory and the minus is to ensure unitarity.)

\begin{rmk}
Besides the Chevalley anti-automorphism $e_i\mapsto -f_i$, $f_i\mapsto -f_i$, $h_0\mapsto h_0$, $h_i\mapsto -h_i$, $k_i\mapsto -k_i$, $u^\pm \to u^\mp$ for $i\in\{1,2\}$, there are a number of involutive automorphisms of $U(\mfa)$ given by
\spl{
& f_i \mapsto e_i , && e_i \mapsto f_i , && h_i \mapsto h_i, && k_i \mapsto k_{j}, && h_0 \mapsto -h_0, && u^\pm \mapsto u^\mp , \\
& f_i \mapsto e_{j} , && e_i \mapsto f_{j} , && h_i \mapsto h_{j}, && k_i \mapsto k_i, && h_0 \mapsto -h_0, && u^\pm \mapsto u^\pm , \\
& f_i \mapsto f_{j} , && e_i \mapsto e_{j} , && h_i \mapsto h_{j}, && k_i \mapsto k_{j}, && h_0 \mapsto h_0, && u^\pm \mapsto u^\mp , \label{aut-Ua}
}
for all $i,j\in\{1,2\}$ such that $i\ne j$. These outer-automorphisms, that form a Klein-four group, are also Hopf algebra outer-automorphisms of $U(\mfa)$. 
\end{rmk}

\begin{rmk} \label{R:2.2}
The relation \eqref{a:Lie} in the algebra $U(\mfa)$ is equivalent to saying that
\eq{
[ e_i , f_j ] = \del_{ij} h_i + \al_i(1-\del_{ij})(u^{+2} - u^{-2}). 
}
We could have started our considerations using this relation as the initial definition, however we chose \eqref{a:Lie} (and \eqref{I0}) to keep our definitions and notation as close as possible to the traditional notation used in \cite{Be3,BL,BSS1,BSSS2,HPT,GH,PTW}). In Section \ref{Sec:5}, where we consider a $q$-deformation of $U(\mfa)$, we use a more natural definition of the central extensions, in such a way slightly deviating from the traditional notation used in \cite{BGM,BK,Ho}. We also find elements $k_i$ to be a good book-keeping notation when considering the highest-weight modules of $U(\mfa)$.

\end{rmk}

\subsection{Typical module} \label{Sec:2.2}

The typical module $K(\la_1,\la_2,\nu)$ is the four-dimensional highest-weight Kac module of $U(\mfa)$ defined as follows: let $v_0\in K(\la_1,\la_2,\nu)$ be the highest-weight vector such that
\eq{
h_i . v_0 = \la_i v_0 , \qu k_i . v_0 = \mu_i v_0 , \qu u^\pm . v_0 = \nu^{\pm1} v_0 , \qu e_i . v_0 = 0 , \qu h_0.v_0 = 0 , 
}
for $i\in\{1,2\}$, where $\la_i,\nu\in\C^\times$ are generic and $\mu_i=\al_i(\nu^2-\nu^{-2})$ due to \eqref{I0}. Set $v_i = f_i.v_0$ and $v_{21} = f_2 f_1 . v_0$. Then $K(\la_1,\la_2,\nu)\cong {\rm span}_\C\{v_0, v_1, v_2, v_{21}\}$ as a vector space. Moreover $h_0.v_i = - v_i$ and $h_0.v_{21} = -2 v_{21}$. This allows us to write the following weight space decomposition:
$$
K(\la_1,\la_2,\nu) = K_0(\la_1,\la_2,\nu) \op K_{-1}(\la_1,\la_2,\nu) \op K_{-2}(\la_1,\la_2,\nu) , 
$$
where $K_0(\la_1,\la_2,\nu)$ and $K_{-2}(\la_1,\la_2,\nu)$ are one-dimensional weight spaces accommodating vectors $v_0$ and $v_{21}$, respectively, and $K_{-1}(\la_1,\la_2,\nu)={\rm span}_\C\{v_1,v_2\}$.

Observe that $f_2.v_1=-f_1.v_2=v_{21}$ and
\spl{
e_1.v_{21} &= \mu_1 v_1 - \la_1 v_2 , \qu e_1.v_1=\la_1 v_0 , \qu e_1.v_2=\mu_1 v_0 , \\
e_2.v_{21} &= \la_2 v_1 - \mu_2 v_2 , \qu e_2.v_1=\mu_2 v_0 , \qu e_2.v_2=\la_2 v_0 . \label{ei.v21}
} 
Introduce linear combinations 
\spl{ \label{v'}
v'_1 &= \mu_1 v_1 - \la_1 v_2	, \qu f'_1 = \mu_1 f_1 - \la_1 f_2, \qu e'_1 = \mu_2 e_1 - \la_1 e_2 , \\
v'_2 &= \la_2 v_1 - \mu_2 v_2 , \qu f'_2 = \la_2 f_1 - \mu_2 f_2, \qu e'_2 = \la_2 e_1 - \mu_1 e_2 .
}
Then 
\spl{
f'_1.v_0 &= v'_1 , \qu f_1.v'_1 = \la_1 v_{21} , \qu f_1.v'_2 = \mu_2 v_{21} , \qu e_1.v'_2 = (\la_1\la_2 - \mu_1\mu_2) v_{0} , \\
f'_2.v_0 &= v'_2,  \qu f_2.v'_1 = \mu_1 v_{21}, \qu f_2.v'_2 = \la_2 v_{21} , \qu e_2.v'_1 = (\mu_1\mu_2-\la_1\la_2) v_{0} .
}
Note that elements $f'_i$, $e'_i$ and $v'_i$ are pairwise linearly independent for generic $\la_i$ and $\nu$. We call the set $\{v_0,v_1,v_2,v_{21}\}$ the {\it up-down} vector space basis and $\{v_0,v'_1,v'_2,v_{21}\}$ the {\it down-up} vector space basis of $K(\la_1,\la_2,\nu)$. The module diagrams for both bases are shown in Figure \ref{Fig:1} (a) and (b).

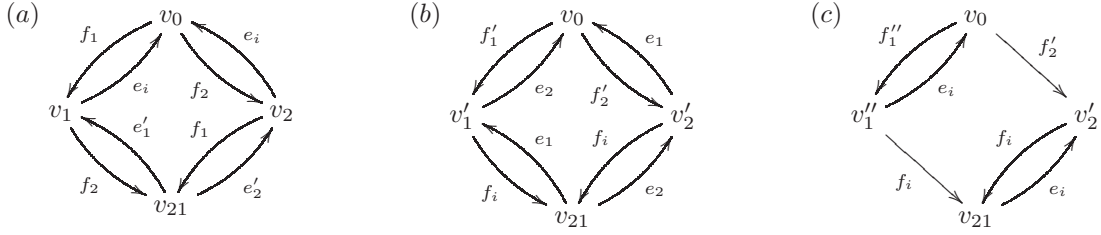
\begin{figure} 
$$(a)
\xymatrix{ 
& v_0 \ar@/_/[ld]_{f_1} \ar@/_/[rd]_{f_2} \\
v_1 \ar@/_/[rd]_{f_2} \ar@/_/[ru]_{e_i}  & & v_2 \ar@/_/[ld]_{f_1} \ar@/_/[lu]_{e_i} \\
& v_{21} \ar@/_/[ru]_{e'_2} \ar@/_/[lu]_{e'_1}
}
\qq\qq (b)
\xymatrix{ 
& v_0 \ar@/_/[ld]_{f'_1} \ar@/_/[rd]_{f'_2} \\
v'_1 \ar@/_/[rd]_{f_i} \ar@/_/[ru]_{e_2}  & & v'_2 \ar@/_/[ld]_{f_i} \ar@/_/[lu]_{e_1} \\
& v_{21} \ar@/_/[ru]_{e_2} \ar@/_/[lu]_{e_1}
}
\qq\qq (c)
\xymatrix{ 
& v_0 \ar@/_/[ld]_{f''_1} \ar[rd]^{f'_2} \\
v''_1 \ar[rd]_{f_i} \ar@/_/[ru]_{e_i}  & & v'_2 \ar@/_/[ld]_{f_i} \\
& v_{21} \ar@/_/[ru]_{e_i} 
}
$$
\caption{Typical module: (a) in up-down basis, (b) in down-up basis, (c) when $\la_1\la_2=\mu_1\mu_2$. } \label{Fig:1}

\end{figure}

\subsection{Atypical module} \label{Sec:2.3}

The atypical module $A(\la_1,\la_2,\nu)$ is the two-dimensional submodule of the typical module $K(\la_1,\la_2,\nu)$ when $\la_1\la_2=\mu_1\mu_2$. Let us show that is indeed true. The aforementioned constrain combined with \eqref{ei.v21} and \eqref{v'} implies that
\eq{ \label{gamma}
v'_1 = \ga^2 v'_2 , \qu f'_1 = \ga^2 f'_2  \qu\text{and}\qu e_1 e_2.v_{21}=0 \qu\text{where}\qu \ga^2 = \frac{\mu_1}{\la_2} = \frac{\la_1}{\mu_2}.
}
Set $v''_1 = \mu_1 v_1 + \la_1 v_2$ and $f''_1 = \mu_1 f_1 + \la_1 f_2$. Then clearly both $v''_1$, $v'_2$ and $f''_1$, $f'_2$ are linearly independent and 
\eq{
f''_1.v_0=v''_1, \qu f'_2.v''_1=-(\la_1\la_2+\mu_1\mu_2)v_{21}, \qu e_1.v''_1=2\la_1\mu_1 v_0, \qu e_2.v''_1=(\la_1\la_2+\mu_1\mu_2)v_0 .
}
The module diagram of $K(\la_1,\la_2,\nu)$ when $\la_1\la_2=\mu_1\mu_2$ is shown in Figure \ref{Fig:1} (c). It is clear from this diagram that $A(\la_1,\la_2,\nu)\cong{\rm span}_\C\{v'_2,v_{21}\}$ as a vector space. Moreover, it follows that 
\eq{
A(\la_1,\la_2,\nu)\cong K(\la_1,\la_2,\nu)/ A(\la_1,\la_2,\nu) .
}

For the purposes of the present paper it will be convenient to choose vector space basis the atypical module to be 
\eq{
A(\la_1,\la_2,\nu) = {\rm span}_\C\{w_0,w_1\}, \label{atypical-mod}
}
where $w_0 = v'_2$ and $w_1 = \ga^{-1} v_{21}$. (This choice of the basis is convenient for obtaining an elegant expression of the $R$-matrix \eqref{R}.) The action of $U(\mfa)$ is given by 
\eq{
h_1.w_j = \ga^2\mu_2\,w_j, \qu h_2.w_j = \ga^{-2}\mu_1\,w_j \qu k_i.w_j = \mu_i\,w_j \qu u^\pm.w_j = \nu^{\pm1} w_j
}
for $i,j\in\{1,2\}$ and
\spl{ \label{mod}
e_i.w_0 &= 0, && f_1.w_0 = \ga \mu_2\, w_1, && f_2.w_0 = \ga^{-1}{\mu_1}\, w_1, \\
f_i.w_1 &= 0, && e_1.w_1 = \ga\, w_0, && e_2.w_1 = \ga^{-1}w_0.
}

A connection with the traditional parametrization of the atypical module in terms of the $x^\pm$-variables used in \cite{BSS1,BSSS2} is given in Appendix~\ref{App:AdS}. 

\begin{rmk}
The module diagrams in Figure \ref{Fig:1} are very similar to those of the $U(\mfsl({1|1}))$ Lie superalgebra. Recall that $\mfsl({1|1})={\rm span}_\C\{f,h,e\}$ as a vector space. The highest-weight module $K(\la)$ of $U(\mfsl({1|1}))$ is the two-dimensional Kac module spanned by vectors $v_0$, $v_1$ such that $e.v_0=0$, $h.v_0=\la v_0$, $f.v_0= v_1$, $e.v_1=\la v_0$. The projective module $P(\la)$ of $U(\mfsl({1|1}))$ is a four-dimensional module spanned by vectors $u_0$, $u_0'$ and $u_{\pm1}$ such that $f.u_0 = u_1$, $e.u_0 = u_{-1}$, $fe.u_0 = u'_{0}$, $h.u_0=\la u_0$, $e.u_{1}=\la u_0-u'_0$, $e.u'_{0}=\la u_{-1}$. Then the diagram in Figure \ref{Fig:1} (c) exactly coincides with the one for $P(\la)$ upon identification $v''_1\to u_0$, $v_0\to u_{1}$, $v'_2\to u'_{0}$, $v_{21}\to u_{-1}$, $e_i\to f$ and $f_i,f'_2,f''_1\to e$. The diagram of $K(\la_1,\la_2,\nu)$ is the completion of the one for $P(\la)$ (and clearly the diagram of $A(\la_1,\la_2,\nu)$ is equivalent to the one of $K(\la)$).

\end{rmk}

\subsection{Tensor product of atypical modules} \label{Sec:2.4}

We show below that the tensor product of two atypical modules is isomorphic to the typical module of $U(\mfa)$. Let $w_i\ot w'{}_{\!\!j}\in A(\la_1,\la_2,\nu)\ot A(\la'_1,\la'_2,\nu')$ with $i,j\in\{0,1\}$. The action of $U(\mfa)$ on vectors $w_i\ot w'{}_{\!\!j}$ is given by
\spl{ \label{f:ww}
\Delta(f_1).(w_0 \ot w'_0) &= \ga\mu_2 \nu' w_1 \ot  w'_0 + (-1)^{p(w_0)} \ga' \mu'_2 \nu^{-1} w_0 \ot  w'_1 , \\
\Delta(f_2).(w_0 \ot  w'_0) &= \ga^{-1}\mu_1 \nu^{\prime-1}\, w_1 \ot  w'_0 + (-1)^{p(w_0)} \ga^{\prime-1} \mu'_1 \nu \, w_0 \ot  w'_1 , \\
\Delta(f_2f_1).(w_0 \ot  w'_0) &= (-1)^{p(w_0)} ( \ga^{-1}\ga' \mu_1 \mu'_2 \nu^{-1} {\nu}^{\prime-1} - \ga \ga^{\prime-1} \mu_2 \mu'_1\nu\nu')\, w_1 \ot  w'_1 
}
and
\spl{ \label{e:ww}
\Delta(e_1).(w_1 \ot  w'_1) &= \ga \nu^{\prime-1}  w_0 \ot  w'_1 - (-1)^{p(w_0)} \ga' \nu \, w_1 \ot  w'_0 , \\
\Delta(e_2).(w_1 \ot  w'_1) &= \ga^{-1} \nu' w_0 \ot  w'_1 - (-1)^{p(w_0)} \ga^{\prime-1}\nu^{-1} w_1 \ot  w'_0 , \\
\Delta(e_1e_2).(w_1 \ot  w'_1) &= (-1)^{p(w_0)}(\ga^{-1}\ga' \nu \nu' - \ga \ga^{\prime-1}\nu^{\prime-1} \nu^{-1} )\, w_0 \ot  w'_0 .
}
Set $\wt v_0=w_0\ot{w}'_0$ and $\wt{v}_i = \Delta(f_i).\wt{v}_0$, $\wt{v}_{21} = \Delta(f_2 f_1).\wt{v}_0$. Then 
\spl{
\Delta(h_1).\wt{v}_0 &= (\ga^{2}\mu_2 + \ga^{\prime2}\mu'_2 )\,\wt{v}_0 = (\la_1 + {\la}'_1 )\,\wt{v}_0 =: \wt{\la}_1\,\wt{v}_0 , \\
\Delta(h_2).\wt{v}_0 &= (\ga^{-2}\mu_1 + \ga^{\prime-2}\mu'_1 )\,\wt{v}_0 = (\la_2 + {\la}'_2 )\,\wt{v}_0 =: \wt{\la}_2\,\wt{v}_0 \\
\Delta(k_i).\wt{v}_0 &= (\mu_i{\nu}^{\prime\dt\mp2} + {\mu}'_i\nu^{\dt\pm2})\,\wt{v}_0 = \al_i(\nu^2\nu^{\prime2}-\nu^{-2}\nu^{\prime-2})\,\wt{v}_0=: \wt{\mu}_i\,\wt{v}_0 ,\label{tilde:la-mu}
}
where the last equality is due to \eqref{I0}. Using the relations above, we find
\spl{
\Delta(e_1).\wt{v}_{21} &= (\mu_1{\nu}^{\prime-2} + \mu'_1 \nu^2)\,\wt{v}_1 - (\ga^2\mu_2+{\ga}^{\prime2}{\mu}'_2)\,\wt{v}_2 = \wt{\mu}_1\,\wt{v}_1 - \wt{\la}_1\,\wt{v}_2 , \\
\Delta(e_2).\wt{v}_{21} &= (\la_2+{\la}'_2)\,\wt{v}_1 - (\mu_2 {\nu}^{\prime2}+\mu'_2 \nu^{-2})\,\wt{v}_2 = \wt{\la}_2\,\wt{v}_1 - \wt{\mu}_2\,\wt{v}_2 ,
}
which compared with \eqref{ei.v21} imply that 
\eq{
K(\wt\la_1,\wt\la_2,\wt\nu) \cong A(\la_1,\la_2,\nu)\ot A(\la'_1,\la'_2,\nu') \label{AA=K}
}
with $\wt\la_i=\la_i+\la'_i$ and $\wt\nu=\nu\nu'$.

\subsection{One-dimensional module} \label{Sec:2.5}

Set 
\eq{
\ms{1} = \ga w_1\ot {w}'_0 + (-1)^{p(w_1)}\, \ga' \nu {\nu}' w_0\ot{w}'_1 . \label{singlet}
}
Let us show that $a.\ms{1}=0$ if $\wt\la_i=\wt\mu_i=0$ for all $a\in U(\mfa)$, $a\ne u^\pm,h_0$. Notice, that it is enough to consider the action of $\Delta(e_i)$ and $\Delta(f_i)$ on $\ms{1}$ only. It follows straightforwardly that $\Delta(e_2).\ms{1}=0$. For $f_2$ we have $\Delta(f_2).\ms{1} = (-1)^{p(w_1)}\ga'\ga\nu( \ga^{\prime-2}\mu'_1 + \ga^{-2}\mu_1)\,w_1 \ot {w}'_1$, which is equal to zero if $\wt\la_2=0$. A simple computation for $f_1$ gives
$$
\Delta(f_1) . \ms{1} = (-1)^{p(w_1)}\ga\ga'\nu( {\mu}'_2 \nu^{-2} + \mu_2 {\nu}^{\prime2})\,w_1 \ot {w}'_1 = (-1)^{p(w_1)}\ga\ga'\nu\wt{\mu}_2\,w_1 \ot {w}'_1 = 0,
$$
if $\wt\mu_2=0$. It remains to consider the action of $e_1$. We have $\Delta(e_1) . \ms{1} = ( \ga^2\nu^{\prime-1} - \ga^{\prime2}\nu^2\nu')\,w_0 \ot w'_0$. Using $\ga^2=\mu_1/\la_2$, $\ga^{\prime2}=\mu'_1/\la'_2$ and $\la'_2=-\la_2$ we find
$$
\Delta(e_1).\ms{1} = \frac{\nu'}{\la_2}\,( \nu^{\prime-2}\mu_1 - \nu^{-2}\mu'_1 )\,w_0 \ot {w}'_0 = \frac{\nu'\wt\mu_1}{\la_2}\,w_0 \ot w'_0 = 0 ,
$$
if $\wt\mu_1=0$. On the other hand, we have $\mu'_2/\mu_2=-\nu^2\nu^{\prime2}$ giving 
$$
\Delta(e_1) . \ms{1} = \frac{1}{\nu' \mu_2}( \ga^2 \mu_2 + \ga^{\prime2} \mu'_2)\,w_0 \ot w'_0 = \frac{\wt\la_1}{\nu' \mu_2}\,w_0 \ot w'_0 = 0,
$$
if $\wt\la_1=0$. Finally, we have that $h_0.\ms{1} \in \C\ms{1}$, and requiring $\Delta(u^\pm).\ms{1}=\ms{1}$ implies $\nu\nu'=1$, which agrees with $\wt\mu_i=0$.

\subsection{R-matrix} \label{Sec:2.6}

Consider the $\Z_2$-graded vector space $\C^{1|1}$. Let $E_{ij}\in\End\,\C^{1|1}$ denote the usual supermatrix units. The algebra $\mf{gl}_{1|1}$ is a $\Z_2$-graded matrix algebra spanned by supermatrices $E_{ij}$ with $i,j\in\{1,2\}$ so that the grading is given by $\deg_2 E_{ij}=p(i)+p(j)$, where $p(1)=0$, $p(2)=1$, and satisfying
$$
[E_{ij}, E_{kl}] = \del_{jk} E_{il} - (-1)^{(i+j)(k+l)}\del_{il} E_{kj} .
$$
Moreover, for any $X,X',Y,Y'\in\mf{gl}_{1|1}$ we have $(X\ot Y)(X'\ot Y') = (-1)^{p(X') p(Y)} XX'\ot YY'$ and $\deg_2(X\ot Y) = \deg_2 X + \deg_2 Y$. In particular, $(E_{ij}\ot E_{kl})(E_{rs}\ot E_{tv}) = (-1)^{(k+l)(r+s)} E_{is} \ot E_{kv}$. 
The graded permutation operator $P\in \End\,(\C^{1|1} \ot \C^{1|1})$ is given by $P = \sum_{1\le i,j \le 2} (-1)^{j-1} E_{ij} \ot E_{ji}$.

Let $\deg_2 v_0 = 0$. Then $\deg_2 w_1=0$ and $\deg_2 w_0=1$. Set $V={\rm span}_\C\{w_1,w_0\}$ and identify $V$ with $\C^{1|1}$ in the natural way. Denote $I=E_{11}+E_{22}$. Then the two dimensional (atypical) representation $\pi : U(\mfa) \to \End(V)$, $a \mapsto \pi(a)$ is given by ({\it c.f.}~\eqref{mod})
\spl{ \label{rep1}
\pi(e_1) &= \ga\, E_{21} , \qu\;\;
\pi(f_1) = \ga\mu_2\, E_{12} , \qu
\pi(f_2) = \ga^{-1}\mu_1\, E_{12} , \qu
\pi(e_2) = \ga^{-1}E_{21} , \\
\pi(h_1) &= \ga^2\mu_2 \, I , \qu
\pi(h_2) = \ga^{-2}\mu_1\, I , \qu
\pi(k_1) = \mu_1 \, I , \qu
\pi(k_2) = \mu_2 \, I , \qu \pi(h_0) = -2E_{11}-E_{22},
}
where $\mu_i = \al_i (\nu^2-\nu^{-2})$. We will refer to the set of parameters $\{\ga,\nu\}$ as the representation {\it labels}. Let $V'={\rm span}_\C\{ w'_1, w'_0\}$ be a copy of $V$. We will denote by $\pi' : U(\mfa) \to \End V'$ the representation with the labels $\{\ga',\nu'\}$.

\begin{prop} \label{P:R}
The $R$-matrix $R(\ga,\nu;\ga',\nu')\in\End(V\ot V')$ intertwining the tensor product of atypical representations $\pi$ and $\pi'$ is given by
\eqa{ \label{R} 
R(\ga,\nu;\ga',\nu') &= \left(\frac{\ga' \nu \nu'}{\ga} - \frac{\ga}{\ga' \nu \nu'}\right) E_{11}\ot E_{11} + \left(\frac{\ga' \nu'}{\ga \nu} - \frac{\ga \nu}{\ga'\nu'}\right) E_{11}\ot E_{22} - \left(\nu^2-\nu^{-2}\right) E_{12}\ot E_{21} \nn\\[.25em]
 & \qu + \left(\nu^{\prime2}-\nu^{\prime-2}\right) E_{21}\ot E_{12} +  \left(\frac{\ga' \nu}{\ga \nu'} - \frac{\ga \nu'}{\ga' \nu}\right) E_{22}\ot E_{11} +  \left(\frac{\ga'}{\ga \nu \nu' } - \frac{\ga \nu \nu'}{\ga'}\right) E_{22}\ot E_{22} .
}
\end{prop}

\begin{proof}
The existence of the $R$-matrix follows from the fact that $V\ot V'$ is an irreducible $U(\mfa)$ module, which is due to \eqref{AA=K}, and application of the Schur's lemma. Thus it only remain to compute it explicitly. This is shown in \mbox{Appendix \ref{App:R}}.
\end{proof}

It is a lengthy but direct computation to check that $R(\ga,\nu;\ga',\nu')$ satisfies the Yang-Baxter equation, namely, set $R_{12} = R(\ga,\nu;\ga',\nu')\ot I$, $R_{23} = I\ot R(\ga',\nu';\ga'',\nu'')$ and $R_{13} = (I\ot P)(R(\ga,\nu;\ga'',\nu'')\ot I)(I\ot P)$ for some (generic) set of labels $\{\ga,\nu;\ga',\nu';\ga'',\nu''\}$, then 
$R_{12} R_{13} R_{23} = R_{23} R_{13} R_{12}$.  

\begin{rmk} [1] \label{R:R}
Let $\nu = e^{i \theta_1}$, $\nu' = e^{i \theta_2}$ and $\ga = e^{i \varLambda}\ga'$.
Then $R$-matrix \eqref{R} can be written in the trigonometric form
\spl{ 
R(\theta_1,\theta_2,\varLambda) &= \sin(\theta_1 \!+\! \theta_2 \!-\! \varLambda)\, E_{11}\ot E_{11} - \sin(\theta_1\!-\!\theta_2 \!+\! \varLambda)\,E_{11}\ot E_{22} - \sin(2\theta_1)\, E_{12}\ot E_{21} \\[.25em]
& \qu + \sin(2\theta_2)\, E_{21}\ot E_{12} + \sin(\theta_1\!-\!\theta_2 \!-\! \varLambda)\, E_{22}\ot E_{11} - \sin(\theta_1\!+\!\theta_2 \!+\! \varLambda)\, E_{22}\ot E_{22} . \label{R:trig} 
}
The unitarity property is
\eq{
R(\theta_1,\theta_2,\varLambda) P R(\theta_2,\theta_1,-\varLambda) P = \frac{1}{2}\,(\cos(2\varLambda) - \cos(2\theta_1+2\theta_2))\,I .
}
In the limit when $\theta_1$, $\theta_2$ and $\varLambda$ are all small, we obtain the rational $R$-matrix of the $\mf{gl}_{1|1}$ Lie superalgebra
\spl{ 
R(\theta,\theta,\varLambda) \approx -(\varLambda I - 2 \theta P) .
}

\noindent (2){\bf .} Let $\bar V={\rm span}_\C\{\bar w_0,\bar w_1\}$ and $\bar V'={\rm span}_\C\{\bar w'_0,\bar w'_1\}$ be such that $\deg_2 \bar w_0= \deg_2 \bar w'_0=0$ and $\deg_2 \bar w_1=\deg_2 \bar w'_1=1$ (i.e.~the $\Z_2$-grading is opposite to that of $V$ and $V'$). Denote the $R$-matrix \eqref{R} by $R_{V,V'}$. Then the $R$-matrices $R_{V,\bar V'}$, $R_{\bar V, V'}$ and $R_{\bar V,\bar V'}$ can be obtained by conjugation $G\,R_{V,V'}\,G$, with a certain matrix $G$ given by the case-by-case basis as follows: 
\spl{
R_{V,\bar V'} \;:\; G = E_{11}\ot(E_{12}+E_{21}) - E_{22} \ot (E_{12} - E_{21} ), \\ 
R_{\bar V,V'} \;:\; G = (E_{12}+E_{21})\ot E_{11} - (E_{12} - E_{21} )\ot E_{22}, \\ 
R_{\bar V,\bar V'} \;:\; G = E_{12} \ot (E_{12}+E_{21}) + E_{21} \ot (E_{12} + E_{21} ) . \label{R:trig-2}
}
(3){\bf .} In the context of the $AdS_3/CFT_2$ duality, parameters $\theta_1$ and $\theta_2$ have an interpretation of the worldsheet momenta $p_{\rm w.s.}$ of individual magnons. The parameter $\varLambda$ can be interpreted as a difference of the pseudorapidity of individual magnons. The traditional notation used in the literature on $AdS_2/CFT_2$ duality can be recovered using relations given in Appendix \ref{App:AdS}. In this context (and using the traditional terminology), the representations $\pi : U(\mfa) \to V$ and $\pi' : U(\mfa) \to V'$ can be viewed as left-moving representations, while $\bar\pi : U(\mfa) \to \bar V$ and $\bar \pi' : U(\mfa) \to \bar V'$ are right-moving representations. With a suitable choice of $\ga$ or $\bar\ga$, which depends on choice of the $AdS_3/CFT_2$ background and additional fluxes, they correspond massless representations. Accordingly, the $R$-matrix $R_{V,V'}$ is the left-left $R$-matrix. Likewise, the ones in \eqref{R:trig-2} are left-right, right-left and right-right $R$-matrices, respectively.  
\end{rmk}


\section{Yangian}

In this section we construct a $u$-deformed Yangian $\mc{Y}(\mfa)$ of Drinfeld-New type \cite{Dr2,ST} having $U(\mfa)$ as its horizontal subalgebra, and obtain an evaluation homomorphism $ ev_\rho : \mcY(\mfa)\to U(\mfa)$.

\subsection{Yangian $\mc{Y}(\mfa)$}

\begin{defn}
The algebra $\mc{Y}(\mfa)$ is the unital associative superalgebra generated by elements $e_{i,r}$, $f_{i,r}$ and $h_{i,r}$, $k_{i,r}$, $h_{0,r}$ with $i,j\in\{1,2\}$, $r\ge0$, and satisfying 
\spl{ 
& [ e_{i,r} , f_{j,s} ] = \del_{ij} h_{i,r+s} + (1-\del_{ij})\, k_{i,r+s} , \qu
 [ h_{0,r} , f_{i,s} ] =  - f_{i,r+s} , \qu [ h_{0,r} , e_{i,s} ] =  e_{i,r+s} . \label{Y(a)}
}
The remaining relations are trivial.
\end{defn}

The $\Z_2$- and $\Z$-grading on $\mcY(\mfa)$ are induced by the ones on $U(\mfa)$, namely $\deg_2(a_{i,r})=\deg_2(a_{i})$ and $\deg(a_{i,r})=\deg(a_{i})$ for all $a_{i,r}\in\mcY(\mfa)$ and $a_i\in U(\mfa)$. The algebra $\mc{Y}(\mfa)$ admits a unique coproduct respecting the $\Z$-grading.
\begin{prop} \label{P:31}
There is a unique homomorphism of algebras $\Delta : \mc{Y}(\mfa) \to \mc{Y}(\mfa) \ot \mc{Y}(\mfa)$ respecting the $\Z$-grading and given by, for $i,j\in\{1,2\}$, $i\ne j$ and $r\ge0$: 
\spl{
\Delta(e_{i,r}) &= e_{i,r} \ot u^{\dt\mp} + u^{\dt\pm} \ot e_{i,r} + \sum_{l=1}^r \left( u^{\dt\pm} h_{i,r-l}\ot e_{i,l-1} + u^{\dt\mp} k_{i,r-l}\ot u^{\dt\mp2} e_{j,l-1} \right) , \\
\Delta(f_{i,r}) &= f_{i,r} \ot u^{\dt\pm} + u^{\dt\mp} \ot f_{i,r} + \sum_{l=1}^r \left( f_{i,r-l} \ot u^{\dt\pm} h_{i,l-1} + u^{\dt\mp2} f_{j,r-l} \ot u^{\dt\mp} k_{j,l-1} \right) , \\
\Delta(h_{i,r}) &= h_{i,r} \ot 1 + 1 \ot h_{i,r} + \sum_{l=1}^r \left( h_{i,r-l}\ot h_{i,l-1} + u^{\dt\mp2} k_{i,r-l}\ot u^{\dt\mp2} k_{j,l-1} \right) , \\
\Delta(k_{i,r}) &= k_{i,r} \ot u^{\dt\mp2} + u^{\dt\pm2} \ot k_{i,r} + \sum_{l=1}^r \left( k_{i,r-l}\ot u^{\dt\mp2} h_{j,l-1} + u^{\dt\pm2} h_{i,r-l}\ot k_{i,l-1} \right) , \\
\Delta(h_{0,r}) &= h_{0,r} \ot 1 + 1 \ot h_{0,r} - \sum_{l=1}^r \left( u^{+} f_{1,r-l} \ot u^{+} e_{1,l-1} + u^{-} f_{2,r-l} \ot u^{-} e_{2,l-1} \right) .  \label{Y(a):cop}
}
\end{prop}

\begin{proof}
Let $\Delta_\eps : \mc{Y}(\mfa) \to \mc{Y}(\mfa) \ot \mc{Y}(\mfa)$ be homomorphism of algebras. By requiring $\Delta_\eps$ to respect the $\Z$-grading we obtain the following ansatz
\spl{ 
\Delta_\eps(e_{i,r}) &= e_{i,r} \ot u^{\dt\mp} + u^{\dt\pm} \ot e_{i,r} + \sum_{l=1}^r \left( \eps_{i,1} u^{\dt\pm} h_{i,r-l}\ot e_{i,l-1} + \eps_{i,2} u^\mp k_{i,r-l}\ot u^{\dt\mp2} e_{j,l-1} \right) , \\
\Delta_\eps(f_{i,r}) &= f_{i,r} \ot u^{\dt\pm} + u^{\dt\mp} \ot f_{i,r} + \sum_{l=1}^r \left( \eps_{i,3} f_{i,r-l} \ot u^{\dt\pm} h_{i,l-1} + \eps_{i,4} u^{\dt\mp2} f_{j r-l} \ot u^{\dt\mp} k_{j,l-1} \right) , \\
\Delta_\eps(h_{i,r}) &= h_{i,r} \ot 1 + 1 \ot h_{i,r} + \sum_{l=1}^r \left( \eps_{i,5} h_{i,r-l}\ot h_{i,l-1} + \eps_{i,6}  u^{\dt\mp2} k_{i,r-l}\ot u^{\dt\mp2} k_{j,l-1} \right) , \\
\Delta_\eps(k_{i,r}) &= k_{i,r} \ot u^{\dt\mp2} + u^{\dt\pm2} \ot k_{i,r} + \sum_{l=1}^r \left( \eps_{i,7} k_{i,r-l}\ot u^{\dt\mp2} h_{j,l-1} + \eps_{i,8} u^{\dt\pm2} h_{i,r-l}\ot k_{i,l-1} \right) , \\
\Delta_\eps(h_{0,r}) &= h_{0,r} \ot 1 + 1 \ot h_{i,r} + \sum_{l=1}^r \left( \eps_{9} u^{+} f_{1,r-l}\ot u^{+} e_{1,l-1} + \eps_{10} u^{-} f_{2,r-l} \ot u^{-} e_{2,l-1} \right) , \\
}
with certain $\eps_{i,k}\in\C$ and $\eps_{9},\eps_{10}\in\C$. Let us focus on generators $e_{i,r}$, $f_{i,r}$, $h_{i,r}$, $k_{i,r}$ first. To find constants $\eps_{i,k}$ we compute the graded commutators
\spl{
{}[\Delta_\eps(e_{i,r}),\Delta_\eps(f_{i,s})] & = h_{i,r+s} \ot 1 + 1 \ot h_{i,r+s} \\
& + \sum_{l=1}^r \left( \eps_{i,1} h_{i,r-l}\ot h_{i,s+l-1} + \eps_{i,2} u^{\dt\mp2} k_{i,r-l}\ot u^{\dt\mp2} k_{j,s+l-1} \right) \\
& + \sum_{l=1}^s \left( \eps_{i,3} h_{i,r+s-l} \ot h_{i,l-1} + \eps_{i,4} \al_{j}\,u^{\dt\mp2} k_{i,r+s-l} \ot u^{\dt\mp2} k_{j,l-1} \right) 
}
and
\spl{
{}[\Delta_\eps(e_{i,r}),\Delta_\eps(f_{j,s})] & = k_{i,r+s} \ot u^{\dt\mp2} + u^{\dt\pm2} \ot k_{i,r+s} \\
& + \sum_{l=1}^r \left( \eps_{i,1} u^{\dt\pm2} h_{i,r-l}\ot k_{i,s+l-1} + \eps_{i,2}\al_i\,  k_{i,r-l}\ot u^{\dt\mp2} h_{j,s+l-1} \right) \\
& + \sum_{l=1}^s \left( \eps_{j,3} k_{i,r+s-l} \ot u^{\dt\mp2} h_{j,l-1} + \eps_{j,4} \al_{i}\,u^{\dt\pm2} h_{i,r+s-l} \ot k_{i,l-1} \right).
}
By requiring $\Delta_\eps$ to be algebra homomorphism and using the expressions above we find $\eps_{i,1} = \eps_{i,3} = \eps_{i,5} = \eps_{j,4} = \eps_{i,8}$ and $\eps_{i,2} = \eps_{i,4} = \eps_{j,3} = \eps_{i,6} = \eps_{i,7}$, where as before, $i\ne j$. Set $\eps_i=\eps_{i,1}$. Then, comparing both sides of $[\Delta(h_{0,r}),\Delta(e_{i,s})]=\Delta(e_{i,s})$ and $[\Delta(h_{0,r}),\Delta(f_{i,s})]=-\Delta(f_{i,s})$, we find that $\eps_9=-\eps_1$ and $\eps_{10}=-\eps_2$. (For completeness, the resulting coproduct is listed Appendix \ref{App:Cop}.)

It remains to show that for any choice of constants $\eps_i$ the coproduct is equivalent up to an automorphism of the Yangian $\mcY(\mfa)$. Consider the map given by
\eq{ \label{omega}
\omega \qu:\qu 
e_{i,r} \mapsto e_{i,r} , \qu f_{i,r} \mapsto \eps_i f_{i,r} , \qu h_{i,r} \mapsto \eps_i h_{i,r} , \qu k_{i,r} \mapsto \eps_{j} k_{i,r} , \qu h_{0,r} \mapsto h_{0,r} .
}
It is easy to see that \eqref{Y(a)} is invariant under the map $\omega$, thus it is an automorphism of $\mcY(\mfa)$. Then $\Delta_{\eps}$ is equivalent to $\Delta$ if $\Delta(a)=(\omega^{-1}\ot\omega^{-1})(\Delta_\eps\circ\omega)(a)$ holds for all $a\in\mcY(\mfa)$. Let us demonstrate that this is true for $a=f_{i,r}$. Using \eqref{Delta:e} we have:
\spl{ \label{omega-f}
& (\omega^{-1}\otimes\omega^{-1})(\Delta_\eps\circ\omega)(f_{i,r}) \\
& \qu = \eps_i\, (\omega^{-1}\otimes\omega^{-1}) \Big( f_{i,r} \ot u^{\dt\pm} + u^{\dt\mp} \ot f_{i,r} + \sum_{l=1}^r \left( \eps_i f_{i,r-l} \ot u^{\dt\pm} h_{i,l-1} + \eps_{j}\, u^{\dt\mp2} f_{j,r-l} \ot u^{\dt\mp} k_{j,l-1} \right) \! \Big) \! \\
& \qu = f_{i,r} \ot u^{\dt\pm} + u^{\dt\mp} \ot f_{i,r} + \sum_{l=1}^r \left( f_{i,r-l} \ot u^{\dt\pm} h_{i,l-1} + u^{\dt\mp2} f_{j,r-l} \ot u^{\dt\mp} k_{j,l-1} \right)  = \Delta(f_{i,r}) .
}
In a similar way one can show this equality is also true for generators $e_{i,r}$, $k_{i,r}$, $h_{i,r}$ and $h_{0,r}$. Thus, without loss of generality, we can set $\eps_i=1$ and obtain \eqref{Y(a):cop} as required.
\end{proof}

Next, we want to write the defining relation \eqref{Y(a)} and coproduct \eqref{Y(a):cop} in terms of generating series (Drinfeld currents). We introduce the following series in $\mc{Y}(\mfa)[[z^{-1}]]$
\eq{ \label{a(z)}
e_i(z)=\sum_{r\ge0} e_{i,r} z^{-r-1} , \qu\! f_i(z)=\sum_{r\ge0} f_{i,r} z^{-r-1} , \qu\! h_i(z)=1+\sum_{r\ge0} h_{i,r} z^{-r-1} , \qu\! k_i(z)=\sum_{r\ge0} k_{i,r} z^{-r-1} , \! 
}
for $i\in\{1,2\}$ and $i=0$ should also be included for $h_i(z)$.
\begin{prop}
The defining relations \eqref{Y(a)} are equivalent to the following identities in $\mcY(\mfa)[[z^{-1},w^{-1}]]$ ($i,j\in\{1,2\}$):
\spl{
& \qq (w-z)\, [e_i(z),f_j(w)] = \del_{ij} (h_i(z)-h_i(w)) + (1-\del_{ij})(k_i(z)-k_i(w)) ,\\
& (w-z)\, [h_0(z),f_i(w)] = f_i(w)-f_i(z) , \qu (w-z)\, [h_0(z),e_i(w)] = e_i(z)-e_i(w) ,  \label{fields}
}
The coproduct \eqref{Y(a):cop} is equivalent to the following identities in $(\mcY(\mfa)\ot \mcY(\mfa))[[z^{-1}]]$   ($i,j\in\{1,2\}$, $i\ne j$): 
\spl{ \label{Y(a):dcop}
\Delta(e_i(z)) &= e_i(z) \ot u^{\dt\mp} + u^{\dt\pm} h_i(z) \ot e_i(z) + u^{\dt\mp} k_i(z) \ot u^{\dt\mp2} e_{j}(z) , \\
\Delta(f_i(z)) &= u^{\dt\mp} \ot f_i(z) + f_i(z) \ot u^{\dt\pm} h_i(z) + u^{\dt\mp2} f_{j}(z) \ot u^{\dt\mp}k_{j}(z) , \\
\Delta(h_i(z)) &= h_i(z) \ot h_i(z) + u^{\dt\mp2} k_i(z) \ot u^{\dt\mp2} k_{j}(z) , \\
\Delta(k_i(z)) &= k_i(z) \ot u^{\dt\mp2} h_{j}(z) + u^{\dt\pm2} h_i(z)\ot k_{i}(z) , \\
\Delta(h_0(z)) &= h_0(z) \ot 1 + 1 \ot h_0(z) - u^{+} f_{1}(z) \ot u^+ e_1(z) - u^{-} f_2(z)\ot u^- e_{2}(z) .
}
\end{prop}

\begin{proof}
To see that \eqref{Y(a)} imply \eqref{fields} we need to multiply \eqref{Y(a)} by $z^{-r-1} w^{-s-1}$. Then sum over $r,s\ge 0$ and use 
$$
(z-w) \sum_{r,s\ge 0,\; r+s=t} z^{-r-1} w^{-s-1}= w^{-t-1} - z^{-t-1} .
$$
This gives \eqref{fields}. Conversely, using the identity above we deduce that 
$$
h_i(z) - h_i(w) = (w-z) \sum_{r,s\ge0} z^{-r-1} w^{-r-1} h_{i,r+s}.
$$
Equating this with $(w-z)[e_i(z),f_i(w)]$ yields $[e_{i,r},f_{i,s}]=h_{i,r+s}$. The remaining relations are obtained similarly. 

To see that \eqref{Y(a):cop} imply \eqref{Y(a):dcop} we only need to multiply \eqref{Y(a):cop} by $u^{-r-1}$ and sum over $r\ge 0$. Conversely, take coefficients at $z^{-r-1}$ with $r\ge 0$ of \eqref{Y(a):dcop}. This reproduces \eqref{Y(a):cop}. 
\end{proof}

Set $H(z)=h_1(z)h_2(z) - k_1(z) k_2(z)$. The current $H(z)$ has a well-defined inverse. Indeed,
\spl{ \label{H:inv}
H(z)^{-1} = \frac{h_1(z)^{-1} h_2(z)^{-1}}{1- k_1(z) k_2(z)h_1(z)^{-1} h_2(z)^{-1}} =  \sum_{l\ge0} (k_1(z) k_2(z))^l(h_1(z)^{-1} h_2(z)^{-1})^{l+1} .
}
Since $h_i(z)$ are invertible, the series on the right hand side are elements in $\mc{Y}(\mfa)[[z^{-1}]]$, and the coefficients of $z^{-r}$ are finite sums of monomials $(k_{1,r_1})^{l_1}(k_{2,r_2})^{l_2}(h_{1,r_3})^{l_3}(h_{2,r_4})^{l_4}$ such that $\sum_{i=1}^4 l_i r_i < r$. We are now ready to define the Hopf algebra structure on $\mc{Y}(\mfa)$.

\begin{thrm} \label{Y:Hopf}
The algebra $\mc{Y}(\mfa)$ has a hopf algebra structure, which is given by the coproduct \eqref{Y(a):dcop}, counit ($i,j\in\{1,2\}$)
\eq{
\veps(e_i(z))=\veps(f_i(z))=\veps(k_i(z))=0, \qq \veps(h_0(z))=\veps(h_i(z))=1,
}
and the antipode defined by $S(u^\pm)=u^{\mp}$ and ($i,j\in\{1,2\}$, $i\ne j$)
\spl{
S(e_i(z)) &= - \left(e_i(z) h_{j}(z) - e_{j}(z) k_i(z)\right)H(z)^{-1}, \qq S(h_i(z)) = h_{j}(z)\, H(z)^{-1}, \\
S(f_i(z)) &= - \left(f_i(z) h_{j}(z) - f_{j}(z) k_{j}(z)\right)H(z)^{-1} , \qq S(k_i(z)) = - k_i(z)\, H(z)^{-1} ,\\
S(h_0(z)) &= 1 - h_0(z) + S(f_1(z)) e_1(z) + S(f_2(z)) e_2(z).
}
\end{thrm}

\begin{proof}
The coproduct is a homomorphism by Proposition \ref{P:31}. The counit follows from \eqref{Y(a)} and \eqref{a(z)}. Thus we only need to check that
\eq{ \label{antipode}
M\circ(S\ot id)\circ\Delta(a_i(z))=\iota\circ\veps(a_i(z)), \qq M\circ(id\ot S)\circ\Delta(a_i(z))=\iota\circ\veps(a_i(z)),
}
for all $a(z)\in \mc{Y}(\mfa)$; here $id$ is the identity map, $M : \mc{Y}(\mfa) \ot \mc{Y}(\mfa) \to \mc{Y}(\mfa)$ denotes the associative multiplication and $\iota : \C \to \mc{Y}(\mfa)$ is the unit map. Let us compute \eqref{antipode} explicitly for all $a(z)$: \medskip

\noindent$\circ$ $h_i(z)$: $\left( h_{j}(z) h_i(z) - k_{i}(z) k_{j}(z)\right) H(z)^{-1} = 1$ for both equalities, \medskip

\noindent$\circ$ $k_i(z)$: $\left(-k_i(z) h_{j}(z) + h_{j}(z) k_i(z) \right) u^{\dt\mp2} H(z)^{-1} = 0$ and $\left(k_i(z) h_{i}(z) - h_{i}(z) k_i(z) \right) u^{\dt\pm2} H(z)^{-1} = 0$, \medskip

\noindent$\circ$ $e_i(z)$: $\left( - \left(e_i(z) h_{j}(z) - e_{j}(z) k_i(z)\right)  + h_{j}(z) e_i(z) - k_i(z) e_{j}(z)\right) u^{\dt\mp} H(z)^{-1} = 0$ and 
\eqn{
\;\;\left(e_i(z) H(z) - h_i(z)\left(e_i(z) h_{j}(z) - e_{j}(z) k_i(z)\right)  - k_i(z) \left(e_{j}(z) h_i(z) - e_i(z) k_{j}(z)\right) \right)u^{\dt\pm} H(z)^{-1}&  \\
= \left(e_i(z) H(z) - e_i(z) h_i(z) h_{j}(z) + e_i(z) k_i(z) k_{j}(z) \right) u^{\dt\pm} H(z)^{-1} &= 0 , 
}

\noindent$\circ$ $f_i(z)$: $\left( f_i(z) H(z) - \left(f_i(z) h_{j}(z) - f_{j}(z) k_{j}(z)\right) h_i(z) - \left(f_{j}(z) h_i(z) - f_i(z) k_i(z)\right) k_{j}(z) \right) u^{\dt\pm} H(z)^{-1} $ \medskip

$ \hspace{5.42cm} = \left( f_i(z) H(z) - f_i(z) \left( h_i(z)  h_{j}(z) - k_i(z) k_{j}(z)\right) \right) u^{\dt\pm} H(z)^{-1} = 0$ \medskip
 
$\qq\;$ and $\left(- \left(f_i(z) h_{j}(z) - f_{j}(z) k_{j}(z)\right) + f_i(z) h_{j}(z) - f_{j}(z) k_{j}(z)\right) u^{\dt\mp} H(z)^{-1} = 0$. 

\medskip

\noindent$\circ$ To compute \eqref{antipode} for $h_0(z)$ we will use the identity
\eqn{
& S(f_1(z)) e_1(z) + S(f_2(z)) e_2(z) \\
& \qq = -(f_1(z) h_2(z) - f_2(z) k_2(z)) H(z)^{-1} e_1(z) - (f_2(z) h_1(z) - f_1(z) k_1(z)) H(z)^{-1} e_2(z) \\
& \qq = -f_1(z) (e_1(z) h_2(z) - e_2(z) k_1(z) ) H(z)^{-1} - f_2(z) ( e_2(z) h_1(z) -  e_1(z) k_2(z)) H(z)^{-1} \\
& \qq = f_1(z) S(e_1(z)) + f_1(z) S(e_1(z)).
}
Then $h_0(z)$: $S(h_0(z)) + h_0(z) - S(f_1(z)) e_1(z) - S(f_2(z)) e_2(z)  = 1$ for both equalities.
\end{proof}

In the previous section we noted that $U(\mfa)$ is a two-parameter family of Hopf algebras. The same is true for $\mc{Y}(\mfa)$. Indeed, by requiring $\Delta(k_i(z))=\Delta^{\rm op}(k_i(z))$ we find ($i,j\in\{1,2\}$, $i\ne j$)
\eq{
k_i(z) \ot ( u^{\dt\mp2} h_{j}(z) - u^{\dt\pm2} h_i(z) ) = (u^{\dt\mp2} h_{j}(z) - u^{\dt\pm2} h_i(z)) \ot k_{i}(z) ,
}
which, together with \eqref{a(z)}, implies that
\eq{ \label{ki(z)}
k_i (z) = \al_i \left(u^2 h_1(z) - u^{-2} h_2(z) \right) z^{-1}.
}
Moreover, it follows that $\Delta(h_i(z))=\Delta^{\rm op}(h_i(z))$ is also true. Finally, the coefficients of $z^{-1}$ in \eqref{ki(z)} reproduce \eqref{I0} as required.

\subsection{Evaluation homomorphism}

Recall that the algebra $U(\mfa)$ is isomorphic to the subalgebra of $\mc{Y}(\mfa)$ generated by the elements $e_{i,0}$, $f_{i,0}$, $h_{0,0}$, $h_{i,0}$, $k_{i,0}$. 

\begin{prop}
There is a homomorphism of algebras given by
\eq{ \label{ev}
ev_\rho \;:\; \mc{Y}(\mfa) \to U(\mfa) , \qu a_{i,r} \mapsto \rho^r a_{i} \qu\text{where}\qu \rho = \frac{u^2 h_1 - u^{-2} h_2}{u^2-u^{-2}} , 
}
for all $a_{i,r}\in \mc{Y}(\mfa)$ and $a_i\in U(\mfa)$ with $i=1,2$, $r\ge0$, $a\in\{e,f,h,k\}$.
\end{prop}

\begin{proof}
It follows from \eqref{Y(a)} and \eqref{a:Lie}, \eqref{I0} that $ev_\rho(a_{i,r}) = \rho^r a_{i}$ for some $\rho\in\C[h_i,k_i,u^\pm]$. Taking the coefficients of $z^{-r-2}$ at both sides of \eqref{ki(z)} we find 
\eq{ \label{kir}
k_{i,r+1} = \al_i \left(u^2 h_{1,r} - u^{-2} h_{2,r} \right) ,
}
for $r\ge0$. Since $k_i = \al_i(u^2-u^{-2})$, the evaluation map applied to \eqref{kir} gives
\eq{
\rho^{r+1} \al_i  (u^2-u^{-2}) = \rho^r \al_i \left(u^2 h_{1} - u^{-2} h_{2} \right) ,
}
which implies \eqref{ev} as required.
\end{proof}

It is a direct computation to check that $R$-matrix \eqref{R} intertwines evaluation representations of $\mcY(a)$, namely 
\eq{
(\pi \ot \pi')(ev_\rho \ot ev_\rho)(\Delta(a))\, R(\ga,\nu;\ga',\nu')=R(\ga,\nu;\ga',\nu')\,(\pi \ot \pi')(ev_\rho \ot ev_\rho)(\Delta^{\rm op}(a))
}
for all~$a\in\mcY(\mfa)$. Also note that evaluation homomorphism for the generating series can be written as
\eq{
ev_\rho \;:\; a_i(z) \mapsto i_z \frac{1}{z-\rho}\, a_i ,
}
where $i_z$ denotes the expansion in the domain $|z|\to\infty$.


\section{The deformed superalgebra $U_q(\C\ltimes\mfsl({1|1})^2\op{\C U^{\pm}})$ and its highest-weight modules} \label{Sec:5}

In this section we present a quantum deformation of the superalgebra considered in Section \ref{Sec:2}. This section follows a similar strategy to the one presented in Section \ref{Sec:2}, namely we start by considering a $q$-deformed universal enveloping superalgebra $U_q(\C \ltimes\mfsl({1|1})^2\ot\C^2)$, and then we obtain $ U_q(\mfa)=U_q(\C \ltimes\mfsl({1|1})^2\op{\C U^{\pm}})$ as the quotient of an extension of the previous superalgebra. (Here $U$ is considered as the $q$-deformed analogue of $u$ in $U(\mfa)=U(\C \ltimes\mfsl({1|1})^2\op{\C u^{\pm}})$.)


\subsection{Algebra}

Let $q\in\C^\times$ be generic (not a root of unity). Set 
$$
[x]_q=\frac{x-x^{-1}}{q-q^{-1}}.
$$

\begin{defn}
The centrally extended superalgebra $U_{q}(\C \ltimes\mfsl({1|1})^2\op\C^2)$ is the unital associative Lie superalgebra generated by elements $E_i$, $F_i$, $K^\pm_0$ and central elements $K^\pm_i$, $L^\pm_i$ with $i,j\in\{1,2\}$ satisfying 
\spl{
& K_0^\pm K_0^\mp=K_i^\pm K_i^\mp=L_i^\pm L_i^\mp=1 , \qu K^+_0 E_i K^-_0 = q E_i, \qu K^-_0 F_i K^+_0 = q F_i \tx{for} i\in\{1,2\},\\[.25em] 
& [ E_i , F_j ] = \del_{ij} \frac{K^{+2}_i - K^{-2}_i}{q-q^{-1}} + \al_i\,(1-\del_{ij})\, \frac{L^{+}_i - L^{-}_i}{q-q^{-1}}  \tx{for} i,j\in\{1,2\}.  \label{Uq:Lie}
}
The remaining relations are trivial. The $\Z_2$-grading is given by $\deg_2(K^\pm_0)=\deg_2(K^\pm_i)=\deg_2(L^\pm_i)=0$ and $\deg_2(E_i)=\deg_2(F_i)=1$.
\end{defn}

\begin{rmk} \label{R:aut-Uqa}
Let $\hbar \in \C^\times$ be generic and set $q=e^{\hbar}$, $K^\pm_i = e^{\pm\hbar H_i/2}$, $L^\pm_i = e^{\pm\hbar J_i}$ with an appropriate defining relations for the new elements $H_i$, $J_i$ implied by \eqref{Uq:Lie}. Then, viewed as a $\C[[\hbar]]$-algebra, $U_{q}(\C \ltimes\mfsl({1|1})^2\op\C^2)$ has outer-automorphism group $GL(2)^2$ acting by
\spl{
& \qq\qq \left(\!\!\begin{array}{c} E_1 \\ E_2 \end{array}\!\!\right) \mapsto A \left(\!\!\begin{array}{c} E_1 \\ E_2 \end{array}\!\!\right) , \qu 
\left(\!\!\begin{array}{c} F_1 \\ F_2 \end{array}\!\!\right) \mapsto B \left(\!\!\begin{array}{c} F_1 \\ F_2 \end{array}\!\!\right) , 
\\
&\left(\!\!\begin{array}{cc} K^\pm_1 & \!\!\al_1(L^{\pm}_1)^{1/2} \\ \al_2 (L^{\pm}_2)^{1/2} & K^\pm_2 \end{array}\!\!\right) \mapsto A \left(\!\!\begin{array}{cc} K^\pm_1 & \!\!\al_1(L^{\pm}_1)^{1/2} \\ \al_2 (L^{\pm}_2)^{1/2} & K^\pm_2 \end{array}\!\!\right) B^t , \label{aut-a}
}
for any $(A,B)\in GL(2)^2$. The elements $K^\pm_0$, which act as outer-automorphisms on the subalgebra $U_{q}(\mfsl({1|1})^2\op\C^2)$, are invariant under the action of $GL(2)^2$.
\end{rmk}

Following similar steps as we did in Section \ref{Sec:2}, we want to enlarge the algebra by central elements $U^\pm$ satisfying $U^\pm U^\mp=1$. We will denote this extended algebra by $U_q(\mfa_0)=U_q(\C \ltimes \mfsl({1|1})^2\op\C^2\op\C U^{\pm})$. The next observation follows straightforwardly.

\begin{prop}
The vector space basis of $ U_q(\mfa_0)$ is given in terms of monomials
\spl{ \label{PBW:q}
(E_2)^{r_2} (E_1)^{r_1} (K^+_0)^{l_0} (K^+_1)^{l_1} (K^+_2)^{l_2} (L^+_1)^{l_3} (L^+_2)^{l_4} (U^+)^{l_5} (E_1)^{s_1} (E_2)^{s_2} 
}
with $r_i,s_i \in \{0,1\}$ and $l_i\in\Z$.
\end{prop}

The monomials \eqref{PBW:q} give a Poincar\'e--Birkhoff--Witt type basis of $ U_q(\mfa_0)$, and $ U_q(\mfa_0)\cong U^-_{q0}.U^0_{q0} .U^+_{q0}$ as vector spaces, where $U^-_{q0}$ and $U^+_{q0}$ are the nilpotent subalgebras generated by elements $F_i$ and $E_i$ with $i=1,2$, respectively, and $U^0_{q0}$ is generated by all the remaining elements. Next we give a remark which follows from analogous Remark \ref{R:Z5} for $U(\mfa_0)$.

\begin{rmk}
The algebra $ U_q(\mfa_0)$ admits a $\Z$-grading given by
\spl{
\deg(K^\pm_0)=\deg(K^\pm_i)=\deg(U^\pm)=0, \qu \deg(E_i)=\dt\pm1, \qu \deg(F_{i})=\dt\mp1 \qu \deg(\al_i) = \dt\pm2.
}
\end{rmk}

Let $I_{q0}$ be the ideal of $U_q(\mfa_0)$ generated by the relations
\eq{ \label{I0q}
L^+_i = K^+_1 K^+_2 U^{\dt\pm2} , \qq L^-_i = K^-_1 K^-_2 U^{\dt\mp2} . 
}
Set $ U_q(\mfa)= U_q(\mfa_0)/I_{q0}$. Then one can define a Hopf algebra structure on $ U_q(\mfa)$ by introducing the coproduct
\spl{ \label{Uq:cop}
\Delta(E_i) &= E_i \ot  U^{\dt\mp}  K_i^- + U^{\dt\pm} K_i^+ \ot E_i , \qu  \Delta(F_i) = F_i \ot U^{\dt\pm} K_i^- + U^{\dt\mp} K_i^+\ot F_i , \qu \Delta(C) = C \ot C ,\!\!
}
and the counit and antipode
\spl{
\veps(E_i)=\veps(F_i)=0, \qu S(E_i) = - E_i, \qu S(F_i) = - F_i, \qu \veps(C)=1, \qu S(C)=C^{-1} ,
}
for $C\in\{K^\pm_0,K^\pm_i,L^\pm_i,U^\pm\}$. \medskip

The coproduct \eqref{Uq:cop} is a homomorphism of algebras $\Delta :  U_q(\mfa) \to  U_q(\mfa) \ot  U_q(\mfa)$. Indeed, for $i,j\in\{1,2\}$ and $i\ne j$ we have
$$
[\Delta(E_i),\Delta(F_{j})] =  \frac{\al_i(L^+_i-L^-_i)}{q-q^{-1}} \ot U^{\dt\mp2} K^{-}_1 K^{-}_2 + U^{\dt\pm2} K^{+}_1 K^{+}_2 \ot \frac{\al_i(L^+_i-L^-_i)}{q-q^{-1}}
$$
and
$$
\Delta([E_i,F_{j}]) = \frac{\al_i (L^+_i \ot L^+_i - L^-_i \ot L^-_i)}{q-q^{-1}} ,
$$
which agree with each other provided \eqref{I0q} holds.

\begin{rmk}
Beside the Chevalley anti-automorphism $E_i\mapsto -F_i$, $F_i\mapsto -E_i$, $K^\pm_0\mapsto K^{\pm}_0$, $K^\pm_i\mapsto -K^{\mp}_i$, $L^\pm_i\mapsto L^\mp_i$, $U^\pm \to U^\mp$ for $i\in\{1,2\}$, there are a number of involutive authomorphisms of $ U_q(\mfa)$ given by ($i,j\in\{1,2\}$, $i\ne j$)
\spl{
& E_i \mapsto F_i , && F_i \mapsto E_i , && K^\pm_i \mapsto K^\pm_i, && L^\pm_i \mapsto L^\pm_{j}, && K^\pm_0 \mapsto K^\mp_0,  && U^\pm \mapsto U^\mp , && \al_i \mapsto \al_{j}, \\
& E_i \mapsto F_{j} , && F_i \mapsto E_{j} , && K^\pm_i \mapsto K^\pm_{j}, && L^\pm_i \mapsto L^\pm_i, && K^\pm_0 \mapsto K^\mp_0, && U^\pm \mapsto U^\pm , && \al_i \mapsto \al_i,\\
& E_i \mapsto E_{j} , && F_i \mapsto F_{j} , && K^\pm_i \mapsto K^\pm_{j}, && L^\pm_i \mapsto L^\pm_{j}, && K^\pm_0 \mapsto K^\pm_0, && U^\pm \mapsto U^\mp, && \al_i \mapsto \al_{j},
}
which form the Klein-four outer-automorphism group of $ U_q(\mfa)$. Note that it is also a group of the Hopf algebra outer-automorphisms of $ U_q(\mfa)$.
\end{rmk}

\subsection{Typical module}

The typical module $K_q(\la_1,\la_2,\nu)$ is the four-dimensional highest-weight Kac module of $ U_q(\mfa)$ defined as follows. Let $v_0$ be the highest-weight vector such that
\eq{ \label{qmod}
K^\pm_i . v_0 = q^{\pm\la_i/2} v_0 , \qu L^\pm_i . v_0 = q^{\pm\mu_i/2} v_0 , \qu U^\pm.v_0 = \nu^\pm v_0, \qu E_i . v_0 = 0 , \qu K^\pm_0 . v_0 = 0
}
for $i\in\{1,2\}$, where $\la_i,\nu\in\C^\times$ are generic and 
\eq{ \label{mumu}
q^{\mu_1} = q^{\frac{\la_1+\la_2}2} \nu^{+2}, \qq q^{\mu_2} = q^{\frac{\la_1+\la_2}2} \nu^{-2}
}
due to \eqref{I0q}. Set $v_i = F_i.v_0$ and $v_{21} = F_2 F_1 . v_0$. Thus $K_q(\la_1,\la_2,\nu)\cong {\rm span}_\C\{v_0, v_1, v_2, v_{21}\}$ as a vector space. Since $K^\pm_0.v_i=q^{-1}v_i$ and $K^\pm_0.v_{21}=q^{-2}v_{21}$ we obtain the following weight space decomposition
$$
K_q(\la_1,\la_2,\nu) = K_{q,0}(\la_1,\la_2,\nu) \op K_{q,-1}(\la_1,\la_2,\nu) \op K_{q,-2}(\la_1,\la_2,\nu) , 
$$
satisfying $v_0 \in K_{q,0}(\la_1,\la_2,\nu)$,\, $v_1, v_2 \in K_{q,-1}(\la_1,\la_2,\nu)$ and $v_{21} \in K_{q,-2}(\la_1,\la_2,\nu)$. We have
\spl{
E_1.v_{21} &= \al_1 [\mu_1]_q v_1 - [\la_1]_q v_2, \qu E_1.v_1=[\la_1]_q v_0 , \qu E_1.v_2=\al_1 [\mu_1]_q v_0 , \\
E_2.v_{21} &= [\la_2]_q v_1 - \al_2 [\mu_2]_q v_2, \qu E_2.v_1= \al_2 [\mu_2]_q v_0 , \qu E_2.v_2=[\la_2]_q v_0 . \label{ei.v21:q}
} 
Define linear combinations
\spl{
v'_1 &= \al_1 [\mu_1]_q v_1 - [\la_1]_q v_2	, \qu F'_1 = \al_1 [\mu_1]_q F_1 - [\la_1]_q F_2, \qu E'_1 = \al_2 [\mu_2]_q E_1 - [\la_1]_q E_2 , \\
v'_2 &= [\la_2]_q v_1 - \al_2 [\mu_2]_q v_2 , \qu F'_2 = [\la_2]_q F_1 - \al_2 [\mu_2]_q F_2, \qu E'_2 = [\la_2]_q E_1 - \al_1 [\mu_1]_q E_2 .
}
Then 
\spl{
F'_1.v_0 &= v'_1 , \qu F_1.v'_1 = [\la_1]_q v_{21} , \qu F_1.v'_2 = \al_2 [\mu_2]_q v_{21} , \qu E_1.v'_2 = \vartheta_- v_{0} , \\
F'_2.v_0 &= v'_2,  \qu F_2.v'_1 = \al_1 [\mu_1]_q v_{21}, \qu F_2.v'_2 = [\la_2]_q v_{21} , \qu E_2.v'_1 = -\vartheta_- v_{0} .
}
where we have introduced a short-hand notation $\vartheta_\pm = [\la_1]_q[\la_2]_q \pm \al_1 \al_2 [\mu_1]_q[\mu_2]_q$. Clearly, elements $F'_i$, $E'_i$ and $v'_i$ are pairwise linearly independent for generic $\la_i$ and $\mu_i$. In the same way as for $U(\mfa)$, we call the set $\{v_0,v_1,v_2,v_{21}\}$ the {\it up-down} vector space basis and $\{v_0,v'_1,v'_2,v_{21}\}$ the {\it down-up} vector space basis of $K_q(\la_1,\la_2,\nu)$. The module diagram for both bases are equivalent to the ones shown in Figure \ref{Fig:1} (a) and (b).

\subsection{Atypical module}

The atypical module $A_q(\la_1,\la_2,\nu)$ is the two-dimensional submodule of $K_q(\la_1,\la_2,\nu)$ when its weights satisfy the relation
\eq{ \label{atypical:q}
[\la_1]_q[\la_2]_q=\al_1 \al_2\, [\mu_1]_q[\mu_2]_q ,
}
giving $\vartheta_-=0$. It follows that
\eq{ \label{gaga}
v'_1 = \ga^2 v'_2 , \qu F'_1 = \ga^2 F'_2 \tx{and} E_1E_2.v_{21}=0 \tx{where} \ga^2=\frac{\al_1 [\mu_1]_q}{[\la_2]_q}=\frac{[\la_1]_q}{\al_2 [\mu_2]_q}.
}
Set $v''_1 = \al_1 [\mu_1]_q v_1 + [\la_1]_q v_2$ and $F''_1 = \al_1 [\mu_1]_q F_1 + [\la_1]_q F_2$. Then clearly both $v''_1$, $v'_2$ and $F''_1$, $F'_2$ are linearly independent and 
\spl{
F''_1.v_0 &= v''_1, \qu E_1.v''_1 = 2\al_1 [\la_1]_q[\mu_1]_q v_0, \qu F'_2.v''_1=- \vartheta_+ v_{21}, \qu E_2.v''_1=\vartheta_+ v_0 .
}
The module diagram of $K_q(\la_1,\la_2,\nu)$ when \eqref{atypical:q} holds is equivalent to the one shown in Figure \ref{Fig:1} (c). Hence $A_q(\la_1,\la_2,\nu)\cong{\rm span}_\C\{v'_2,v_{21}\}$ as a vector space, and 
\eq{
A_q(\la_1,\la_2,\nu)\cong K_q(\la_1,\la_2,\nu)/ A_q(\la_1,\la_2,\nu).
}

As before, it will be convenient to choose the vector space basis the atypical module to be
\eq{
A_q(\la_1,\la_2,\nu) = {\rm span}_\C\{w_0,w_1\},
}
where $w_0 = v'_2$ and $w_1 = \ga^{-1} v_{21}$. The action of $ U_q(\mfa)$ is given by 
\eq{
K^\pm_i.w_j = q^{\pm \la_i/2} w_j, \qu L^\pm_i.w_j = q^{\pm\mu_i/2} w_j, \qu U^\pm.w_j = \nu^{\pm1} w_j, \qu K^\pm_0.w_j = q^{\mp(1+j)} w_j
}
for $i,j\in\{1,2\}$ and
\spl{ \label{qmod}
E_i.w_0 &= 0, && F_1.w_0 = \al_2\ga [\mu_2]_q\, w_1, && F_2.w_0 = \al_1 \ga^{-1}[\mu_1]_q\, w_1, \\
F_i.w_1 &= 0, && E_1.w_1 = \ga\, w_0, && E_2.w_1 = \ga^{-1}w_0.
}

A connection with the traditional deformed parametrization of the atypical module in terms of the $x^\pm$ variables
used in \cite{BGM,Ho} is given in Appendix \ref{App:q-AdS}.

\subsection{Tensor product of atypical modules}

Let $w_i \ot w'_j \in A_q(\la_1,\la_2,\nu)\ot A_q(\la'_1,\la'_2,\nu')$ with $i,j\in\{0,1\}$.
A direct computation shows that the action of $ U_q(\mfa)$ on vectors $w_i \ot w'_j$ is given by
\spl{ \label{f:ww:q}
\Delta(F_1).(w_0 \ot  w'_0) &= \al_2 \ga \nu' q^{-\frac{\la'_1}2}[\mu_2]_q  \, w_1 \ot  w'_0 + (-1)^{p(w_0)} \al_2 \ga'\nu^{-1} q^{\frac{\la_1}2} [\mu'_2]_q \, w_0 \ot  w'_1 , \\
\Delta(F_2).(w_0 \ot  w'_0) &= {\nu}^{\prime-1}q^{-\frac{\la'_2}2}[\la_2]_q \, w_1 \ot  w'_0 + (-1)^{p(w_0)} \nu\,q^{\frac{\la_2}2} [\la'_2]_q\, w_0 \ot  w'_1 , \\
\Delta(F_2F_1).(w_0 \ot  w'_0) &= \al_2 {\ga\ga'}(-1)^{p(w_0)} \big((\nu\nu')^{-1} q^{\frac{\la_1-\la'_2}2} [\la_2]_q [\mu'_2]_q - \nu\nu' q^{\frac{\la_2-\la'_1}2} [\mu_2]_q [\la'_2]_q \big)\, w_1 \ot  w'_1 
}
and
\spl{ \label{e:ww:q}
\Delta(E_1).(w_1 \ot  w'_1) &= \ga\nu^{\prime-1} q^{-\frac{\la'_1}2} \, w_0 \ot  w'_1 - (-1)^{p(w_0)} \ga' \nu q^{\frac{\la_1}2}\, w_1 \ot  w'_0 , \\
\Delta(E_2).(w_1 \ot  w'_1) &= \ga^{-1}\bar{\nu}q^{-\frac{\la'_2}2} \, w_0 \ot  w'_1 - (-1)^{p(w_0)}\ga^{\prime-1}\nu^{-1}q^{\frac{\la_2}2} \, w_1 \ot  w'_0 , \\
\Delta(E_1E_2).(w_1 \ot  w'_1) &= (-1)^{p(w_0)} \big (\ga^{-1}\ga' \nu \bar{\nu}q^{\frac{\la_1-\la'_2}2} - \ga\ga^{\prime-1} \nu^{\prime-1}\nu^{-1}q^{\frac{\la_2-\la'_1}2} \big)\, w_0 \ot  w'_0 .
}
Set $\wt v_0=w_0\ot w'_0$ and $\wt{v}_i = \Delta(F_i).\wt{v}_0$, $\wt{v}_{12} = \Delta(F_1F_2).\wt{v}_0$. Then 
\spl{
\Delta(K^\pm_i).\wt{v}_0 &= q^{\pm\frac{\la_i+\la'_i}2}\,\wt{v}_0 =: q^{\pm\frac{\wt{\la}_i}2}\,\wt{v}_0 , \qq	
\Delta(L^\pm_i).\wt{v}_0 = q^{\pm\frac{\mu_i+\mu'_i}2}\,  \wt{v}_0 =: q^{\pm\frac{\wt{\mu}_i}2}\, \wt{v}_0 . \label{tilde:la-mu:q}
}
Using the relation $q^{\mu_i}=q^{\frac{\la_1+\la_2}2}\nu^{\dt\pm2}$ and the expressions above we find
\spl{
\Delta(E_1).\wt{v}_{21} &= \al_1\big({\nu}^{\prime-2}q^{-\frac{\la'_1+\la'_2}2}[\mu_1]_q + \nu^2 q^{\frac{\la_1+\la_2}2}[\mu'_1]_q\big)\,\wt{v}_1 - \al_2(\ga q^{-\la'_1}[\mu_2]_q+\ga'q^{\la_1}[\mu'_2]_q)\,\wt{v}_2 \\
&= \al_1\big(q^{-\mu'_1}[\mu_1]_q + q^{\mu_1}[\mu'_1]_q\big)\,\wt{v}_1 - \big(q^{-\la'_1}[\la_1]_q + q^{\la_1}[\la'_1]_q\big)\,\wt{v}_2 = \al_1[\wt{\mu}_1]_q\,\wt{v}_1 - [\wt{\la}_1]_q\,\wt{v}_2 , \\
\Delta(E_2).\wt{v}_{21} &= [\la_2+\la'_2]_q\,\wt{v}_1 - \al_2\big(\nu^{\prime2} q^{-\frac{\la'_1+\la'_2}2}[\mu_2]_q + \nu^{-2} q^{\frac{\la_1+\la_2}2} [\mu'_2]_q\big) \,\wt{v}_2 \\
&= [\la_2+\la'_2]_q\,\wt{v}_1 - \al_2\big(q^{-\mu'_2}[\mu_2]_q + q^{\mu_2} [\mu'_2]_q\big) \,\wt{v}_2 = [\wt{\la}_2]_q\,\wt{v}_1 - \al_2[\wt{\mu}_2]_q\,\wt{v}_2 ,
}
which compared with \eqref{ei.v21:q} imply that ({\it c.f.}~\eqref{AA=K})
\eq{
K_q(\wt\la_1,\wt\la_2,\wt\nu) \cong A_q(\la_1,\la_2,\nu)\ot A_q(\la'_1,\la'_2,\nu')
}
with $\wt\la_i=\la_i+\la'_i$ and $\wt\nu=\nu\nu'$. Moreover, $q^{\wt\mu_i}=q^{(\wt\la_1+\wt\la_2)/2} \wt\nu^{\,2}$.

\subsection{One-dimensional module} \label{Sec:5.4}

Set 
\eq{ \label{1q}
\ms{1}_q = \ga w_1\ot w'_0 + (-1)^{p(w_1)}\, q^{-\wt\la_2/2}\ga' \nu \nu' w_0\ot w'_1 .
}
Its clear that $K^\pm_0.\ms{1}_q \in \C \ms{1}_q$. We need to show that $a.\ms{1}_q = 0$ if $\wt\la_i=\wt\mu_i=0$ for all $a\in U_q(\mfa)$, $a\ne U^\pm, K^\pm_0$. It follows straightforwardly that $\Delta(E_2).\ms{1}_q=0$. For $F_1$ and $F_2$ we have 
\eqn{
\Delta(F_2).\ms{1}_q &= (-1)^{p(w_1)}\ga \ga'\nu q^{-\frac{\la_2}2}\big( q^{\la_2}[\la'_2]_q + q^{-\la'_2}[\la_2]_q\big)\,w_1 \ot w'_1 = (-1)^{p(w_1)}\ga\ga'\nu q^{-\frac{\la_2}2} [\wt\la_2]_q\,w_1 \ot w'_1, \\
\Delta(F_1) . \ms{1}_q &= \al_2(-1)^{p(w_1)}\ga\ga'\big( \nu^{-1} q^{\frac{\la_1}2}[\mu'_2]_q  + \nu\nu^{\prime2} q^{-\frac{\la'_2+\la'_1}2} [\mu_2]_q\big) \,w_1 \ot w'_1 \\
& = \al_2(-1)^{p(w_1)}\ga\ga'\nu q^{-\frac{\la_2}2}\big( \nu^{-2} q^{\frac{\la_1+\la_2}2}[\mu'_2]_q  + \nu^{\prime2} q^{-\frac{\la'_1+\la'_2}2} [\mu_2]_q\big) \,w_1 \ot w'_1 \\
& = \al_2(-1)^{p(w_1)}\ga\ga'\nu q^{-\frac{\la_2}2}[\wt\mu_2]_q \,w_1 \ot w'_1 , 
}
which are zero only if $\wt\la_2=\wt\mu_2=0$. Now consider the action of $E_1$: 
$$
\Delta(E_1) . \ms{1}_q = \big( \ga^2 \nu^{\prime-1} q^{-\frac{\la'_1}2} - \ga^{\prime2}\nu^2\nu' q^{\frac{\la_1}2}\big)\,w_0 \ot w'_0.
$$
Using $\ga^2=\al_1[\mu_1]_q/[\la_2]_q$, $\ga^{\prime2}=\al_1[\mu'_1]_q/[\la'_2]_q$ and $\la'_2=-\la_2$ we obtain
$$
\Delta(E_1).\ms{1}_q = \al_1\nu' \,\frac{q^{\frac{\la'_2}2}}{[\la_2]_q}\,\big( \nu^{\prime-2} q^{-\frac{\la'_1+\la'_2}2}[\mu_1]_q + \nu^2 q^{\frac{\la_1+\la_2}2} [\mu'_1]_q \big)\,w_0 \ot w'_0  = \al_1\nu' q^{\frac{\la'_2}2}\,\frac{[\wt\mu_1]_q}{[\la_2]_q}\,w_0 \ot w'_0 ,
$$
which is zero if $\wt\mu_1=0$. Finally, using $\ga^2=\al_2^{-1}[\la_1]_q/[\mu_2]_q$, $\ga^{\prime2}=\al_2^{-1}[\la'_1]_q/[\mu'_2]_q$ and $\la'_2=-\la_2$, $\mu'_2=-\mu_2$, we get
\eqn{
\Delta(E_1).\ms{1}_q &= \frac{\nu'}{\al_2[\mu_2]_q} \big( \nu^{\prime-2}q^{-\frac{\la'_1}2}[\la_1]_q + \nu^2q^{\frac{\la_1}2}[\la'_1]_q \big) w_0 \ot w'_0 \\
& = \frac{\nu'}{\al_2[\mu_2]_q} \big( q^{\mu'_2}q^{-\la'_1-\la'_2/2}[\la_1]_q + q^{-\mu_2} q^{\la_1+\la_2/2}[\la'_1]_q \big) w_0 \ot w'_0 
= \frac{\nu' q^{-\mu_2+\la_2/2}}{\al_2[\mu_2]_q} [\wt\la_1]_q \, w_0 \ot w'_0 ,
}
which gives zero if $\wt\la_1=0$. Now $\wt\la_i=\wt\mu_i=0$ and requiring $\Delta(U^\pm).\ms{1}_q = \ms{1}_q$ implies $\wt\nu=1$.

\subsection{$R$-matrix}

We use the same notation as in Section \ref{Sec:2.6}. The two dimensional (atypical) representation $\pi :  U_q(\mfa) \to \End(V)$, $a \mapsto \pi(a)$ is given by ({\it c.f.}~\eqref{qmod})
\spl{ \label{rep1:q}
\pi(E_1) &= \ga\, E_{21} , \qu
\pi(F_1) = \al_2\ga [\mu_2]_q\, E_{12} , \qu
\pi(K^\pm_i) = q^{\pm\la_i/2}\, I , \qu \pi(U^\pm) = \nu^\pm I ,\\
\pi(E_2) &= \ga^{-1} E_{21} , \qu
\pi(F_2) = \al_1\ga^{-1}[\mu_1]_q\, E_{12} , \qu
\pi(L^\pm_i) = q^{\pm\mu_i/2}\, I , \qu \pi(K^\pm_0) = q^{\mp2}E_{11}-q^{\mp}E_{22}.
}
Likewise, we denote by $\pi' :  U_q(\mfa) \to \End V'$ the representation with labels $\{\ga',\nu'\}$.

\begin{prop} \label{P:Rq}
The $R$-matrix $R_q(\ga,\nu;\ga',\nu')\in\End(V\ot V')$ intertwining the tensor product of atypical representations $\pi$ and $\pi'$ is given by 
\eqa{ \label{Rq}
R_q(\ga,\nu;\ga',\nu') &= \Big(q^{\frac{\la_1-{\la'}_2}{2}} \frac{{\ga'} \nu{\nu'}}{\ga}  - q^{\frac{\la_2-{\la'}_1}{2}} \frac{\ga}{\ga' \nu \nu'} \Big) E_{11}\ot E_{11} 
+ \Big(q^{-\frac{\la_1+{\la'}_2}{2}}\frac{{\ga'}\nu'}{\ga \nu}  - q^{-\frac{{\la'}_1+\la_2}{2}} \frac{\ga \nu}{\ga'\nu'}  \Big) E_{11}\ot E_{22} \el[.25em]
& - \big(q^{\frac{\la_1-\la_2}2} \nu^2 - q^{-\frac{\la_1-\la_2}2}\nu^{-2}\big) E_{12}\ot E_{21} 
+ \big(q^{\frac{\la'_1-\la'_2}2} \nu^{\prime2} - q^{-\frac{\la'_1-\la'_2}2}\nu^{\prime-2} \big) E_{21}\ot E_{12} \el[.5em]
& + \Big( q^{\frac{\la_1+\la'_2}{2}} \frac{\ga' \nu}{\ga \nu'} - q^{\frac{\la'_1+\la_2}{2}} \frac{\ga \nu'}{\ga'\nu}\Big) E_{22}\ot E_{11} 
+ \Big( q^{-\frac{\la_1-\la'_2}{2}}\frac{\ga'}{\ga \nu\nu'} - q^{-\frac{\la_2-\la'_1}{2}} \frac{\ga \nu \nu'}{\ga'} \Big) E_{22}\ot E_{22} .
}
\end{prop}

\begin{proof}
The proof follows by the same arguments as those in the proof of Proposition \ref{P:R}. Explicit calculations are given in Appendix \ref{App:Rq}.
\end{proof}

Checking that the deformed $R$-matrix \eqref{Rq} satisfies the Yang-Baxter equation uses the same method described just below Proposition \ref{P:R}. 

\begin{rmk}
[1] Note that $\la_i=\la_i(\ga,\nu)$ and $\la'_i=\la'_i(\ga',\nu')$ in \eqref{Rq} due to \eqref{mumu} and \eqref{atypical:q}. 

\noindent (2){\bf.} Taking the $q\to1$ limit \eqref{Rq} coincides with \eqref{R}.

\noindent (3){\bf.} All what was said in Remark \ref{R:R} (2) and (3) is also true in the deformed case. The traditional deformed notation is briefly discussed in Appendix \ref{App:q-AdS}.
\end{rmk}


\section{Affinization}

Wes affinize the algebra $ U_q(\mfa)$ by doubling its nodes. Our approach is inspired by a similar affinization presented in \cite{BGM}. In this section we will use the additional notation $(i)=(-1)^{i-1}$, which will appear in the powers of the central elements only.

\begin{defn}
The quantum affine algebra $U_{q}(\wh\mfa)$ is the unital associative Lie superalgebra generated by elements $E_i$, $F_i$, $K^\pm_0$ and central elements $K^\pm_i$, $U^{\pm}$, $V^{\pm}$ with $1\le i\le 4$ satisfying
\spl{ \label{Uqa:Lie}
& K_0^\pm K_0^\mp=K_i^\pm K_i^\mp=U^\pm U^\mp=V^\pm V^\mp=1, \qu\! K^+_0 E_i K^-_0 = q E_i , \qu\! K^-_0 F_i K^+_0 = q F_i \!\tx{for}\! 1\le i \le 4, \\[.25em]
& [E_i, F_j] = \del_{ij}\, \frac{K^{+2}_i - K^{-2}_i}{q-q^{-1}} + \al_i (1-\del_{ij})\, \frac{L^+_{i} - L^-_{i}}{q-q^{-1}} \tx{for} i,j \in\{1,2\} \;\text{or}\;\, \{3,4\} , \\[0.25em] 
& [[E_3,F_2],[E_4,F_1]] = \frac{K^{+} - K^{-}}{q-q^{-1}}, \qu [[E_i,F_{i+2}],[E_{j+2},F_{j}]] = \frac{L^+_{i} L^+_{j+2} - L^-_{i} L^-_{j+2}}{q-q^{-1}} \tx{and} \\[0.25em]
& [E_i,F_{j+2}] = \al_i\,\frac{U^{+}V^{+} (K^{+}_i K^{+}_{j+2})^{\sqb{i}} - U^{-}V^{-} (K^{-}_i K^{-}_{j+2})^{\sqb{i}}}{q-q^{-1}}  \tx{for} i,j \in\{1,2\}, \; i\ne j ,
}
where 
\eq{ \label{KL}
K^\pm= \prod_{1\le l \le 4} K^\pm_l, \qu L^{\pm}_i = (U^{\pm2})^{\sqb{i}} K^\pm_1 K^\pm_2 , \qu L^{\pm}_j = (V^{\pm2})^{\sqb{j}} K^\pm_3 K^\pm_4 \tx{for} i\in\{1,2\}, \; j\in\{3,4\}.
}
The remaining relations are trivial. The $\Z_2$-grading is given by $\deg_2(K^\pm_i)=\deg_2(U^{\pm})=\deg_2(V^{\pm})=0$ and $\deg_2(E_i)=\deg_2(F_i)=1$.
\end{defn}

Notice that the affine extension is such that elements with $i=3,4$ together with $V^\pm$ generate a Hopf subalgebra of $U_q(\widehat\mfa)$ isomorphic to the subalgbra of $U_q(\mfa)$ generated by all its elements except $K^\pm_0$. We will refer to the relations in the third line of \eqref{Uqa:Lie} as the quantum Serre relations and to the relation in the fourth line as the compatibility relation. The choice of these additional relations will be explained a little bit further.

\begin{rmk}
Assuming $i\in\{1,2,3,4\}$ and $j\in\{1,2\}$ the $\Z$-grading on $U_{q}(\wh\mfa_0)$ is
\spl{
\deg(E_{2j-1})=\deg(F_{2j})=1, \qu &\deg(E_{2j})=\deg(F_{2j-1})=-1 ,  \\
\deg(K^\pm_0)=\deg(K^\pm_i)=\deg(U^{\pm})=\deg(V^{\pm})=0 , \qu &\deg(\al_{2j-1})=2, \qu \deg(\al_{2j})=-2.
}
\end{rmk}

We can define a Hopf algebra structure on $ U_q(\widehat\mfa)$ as follows.

\begin{thrm}
The algebra $ U_q(\widehat\mfa)$ has a Hopf algebra structure given by the coproduct $\Delta(C) = C \ot C$ and
\spl{ \label{Uqa:cop}
\Delta(E_i) &= E_i \ot U^{{-}\sqb{i}} K^-_i + U^{+\sqb{i}} K^+_i \ot E_i , \qu \Delta(F_i) = F_i \ot U^{+\sqb{i}} K^-_i + U^{-\sqb{i}} K^+_i \ot F_i \qu\text{for } i=1,2,  \!\!\! \\
\Delta(E_i) &= E_i \ot V^{{-}\sqb{i}} K^-_i + V^{+\sqb{i}} K^+_i \ot E_i , \qu \Delta(F_i) = F_i \ot V^{+\sqb{i}} K^-_i + V^{-\sqb{i}} K^+_i \ot F_i \qu\text{for } i=3,4 , \!\!\!
}
counit and antipode
\eq{
\veps(E_i)=\veps(F_i)=0, \qu S(E_i)=-E_i, \qu S(F_i)=-F_i, \qu \veps(C)=1, \qu S(C)=C^{-1}
}
for $C\in\{K^\pm_0,K^\pm_i,U^\pm,V^\pm\}$.
\end{thrm}

\begin{proof}

The proof follows by a direct computation. \qedhere

\end{proof}

\subsection{Evaluation homomorphism} \label{Sec:6.1}

\begin{prop}
There exists a homomorphisms of algebras $ U_q(\widehat\mfa) \to  U_q(\mfa)$ given by
\eq{ \label{evq}
ev_\rho \;\;:\;\; \begin{cases}
E_{i+2} \mapsto -\rho^{-1}(L^+_{j} - L^-_{j}) E_{i} , \qq K^\pm_{i+2} \mapsto K^\mp_{i}, \qu \al_{i+2} \mapsto \al_{i} , \hspace{2.35cm} \\
F_{i+2} \mapsto \rho\, (L^+_{j} - L^-_{j})^{-1} F_{i} , \qq V^\pm \mapsto U^\pm ,  
\end{cases}
}
for $i,j\in\{1,2\}$ and $i\ne j$, where
\spl{ \label{rhoq}
\rho = U^{2}K^+_1 K^-_2 - U^{-2}K^-_1 K^+_2 .
}
\end{prop}

\begin{proof}
It is easy to see that $ev_\rho ([E_{i+2} , F_{i+2}])=[ev_\rho(E_{i+2}),ev_\rho(F_{i+2})]$ and the same is true for the quantum Serre relations, since $ev_\rho(K^\pm)=1$ and $ev_\rho(L^\pm_{i,i+2}L^\pm_{j+2,j}) =1$, which can be deduced from \eqref{KL} (and is true only if $ev_\rho(V^\pm)=U^\pm$ and $ev_\rho(K^{\pm}_{3} K^{\pm}_{4})=K^\mp_{1} K^\mp_{2}$). 
Now consider the compatibility relation. We have
\eq{
ev_\rho([E_i,F_{j+2}])=\al_i\,\frac{U^{2} K^{+\sqb{i}}_i K^{-\sqb{i}}_{j} - U^{-2} K^{-\sqb{i}}_i K^{+\sqb{i}}_{j}}{q-q^{-1}} = \al_i\,\frac{\rho}{q-q^{-1}} 
}
and
\eq{
[ev_\rho(E_i),ev_\rho(F_{j+2})]= \frac{\rho}{L^+_i-L^-_i}\,[E_i, F_{j}] = \al_i\,\frac{\rho}{q-q^{-1}} .
}
Finally,
\eq{
ev_\rho([E_{i+2},F_{j+2}]) = \al_i\, \frac{U^{+2\sqb{i}} K^{-}_i K^{-}_{j} - U^{-\sqb{i}}K^{+}_i K^{+}_{j}}{q-q^{-1}} =\al_i\, \frac{ L^-_{j} - L^+_{j}}{q-q^{-1}} 
}
and
\eq{
[ev_\rho(E_{i+2}),ev_\rho(F_{j+2})] = -\frac{L^+_{j} - L^-_{j}}{L^+_{i} - L^-_{i}}\,[E_i, F_{j}] = \al_i\,\frac{L^-_{j} - L^+_{j}}{q-q^{-1}} ,
}
which agree with each other, as required.
\end{proof}

It is a direct computation to check that $R$-matrix \eqref{Rq} intertwines evaluation representations of $U_q(\widehat\mfa)$, namely
\eq{
(\pi \ot \pi')(ev_\rho \ot ev_\rho)(\Delta(a))\, R_q(\ga,\nu;\ga',\nu')=R_q(\ga,\nu;\ga',\nu')\,(\pi \ot \pi')(ev_\rho \ot ev_\rho)(\Delta^{\rm op}(a))
}
for all~$a\in U_q(\widehat\mfa)$.

\begin{rmk}
[1] The affinization of $ U_q(\mfa)$ presented above is unique up to an isomorphism. For example, one could choose the additional relations in \eqref{Uqa:Lie} to be
\eqn{
& [[E_3,F_1],[E_4,F_2]] = \frac{K^{+} - K^{-}}{q-q^{-1}}, \qu [[E_i,F_{j+2}],[E_{i+2},F_{\tilde\imath}]] = \frac{L^+_{i} L^+_{i+2} - L^-_{i} L^-_{i+2}}{q-q^{-1}} \tx{and} \\[0.25em]
& [E_i,F_{i+2}] = \al_i\,\frac{U^{+}V^{-} K^{+\sqb{i}}_i K^{+\sqb{i}}_{i+2} - U^{-}V^{+} K^{-\sqb{i}}_i K^{-\sqb{i}}_{j+2}}{q-q^{-1}}  \tx{for} i,j \in\{1,2\} , \; i\ne j ,
}
leading to
\spl{
& ev_\rho \;\;:\;\;\begin{cases}
E_{i+2} \mapsto -\rho^{-1}\,(L^+_{i}-L^-_{i}) E_{j} , \qq K^\pm_{i+2} \mapsto K^\mp_{j}, \qu \al_{i+2} \mapsto \al_{j} ,\\
F_{i+2} \mapsto \rho (L^+_{i}-L^-_{i})^{-1} F_{j} , \qq\; V^\pm \mapsto U^\mp .
\end{cases}
}

(2){\bf.} Let $\beta\in\C$ be such that $\beta^2=1$. Then one could substitute the map $V^\pm \mapsto U^\pm$ by $V^\pm \mapsto \beta U^\pm$ in \eqref{evq}, since $U^\pm$ only appear squared in the algebra $ U_q(\mfa)$ (via \eqref{I0q}).
\end{rmk}


\enlargethispage{0.25em}

\section{Conclusions and outlook}

In this paper we have demonstrated novel algebraic structures that are inspired by the $AdS_3/CFT_2$ duality. The main results are the $u$-deformed Yangian $\mc{Y}(\mfa)$ presented in Section 3, the double-deformed quantum affine algebra $ U_q(\widehat\mfa)$ presented in Section 5. The main goal of this study was to construct infinite dimensional deformed superalgebras and double-deformed superalgebras, and obtain their evaluation modules that could further be studied using similar methods to the ones in \cite{HZh,RZh1,RZh2}. We also showed that the $R$-matrix of $U(\mfa)$ can be written in an elegant trigonometric form. (However we were unable to find an elegant parametrization of the deformed $R$-matrix; we leave this question for a further study.) This allow us to study the spectral problem of the $AdS_3/CFT_2$ duality, addressed in \cite{BDS,BSS2,BSZ,SS}, using algebraic methods developed in \cite{BFLMS,BR,MaMe,RS} in the Yangian case, and in \cite{HJ,FH} in the quantum deformed case. We will address these questions in a future publication.

There is a number of important recent developments related the integrability of the $AdS_3/CFT_2$ duality for various backgrounds with mixed NS-NS and R-R fluxes \cite{BSSS1,BSSS2,BSSST1,HPT,HRT,HT,LSSS}. In many cases the underlying symmetry of the worldsheet scattering is the centrally extended $\mfsl({1|1})^2$ or $\mfsl({1|1})^4$ superalgebra, and the complete worldsheet $R$-matrix has a certain block structure, with building blocks in many instances being the four-dimensional $\mfsl({1|1})^2\op\C u^\pm$--symmetric $R$-matrices given by \eqref{R:trig} and \eqref{R:trig-2}, that describe scattering of appropriate species of the worldsheet magnons. The trigonometric parametrization of the $R$-matrix reveals some properties of the scattering that are not that obvious in the traditional notation (which uses the $x^\pm$ variables introduced in \cite{Be1}). Let us illustrate this. In terms of the terminology used in \cite{BSSS2}, the eigenvalues of the magnon Hamiltonian $H=h_1+h_2$ and angular momentum $M=h_1-h_2$ operators are
$$
H(p) = -4h \sin(2\varLambda_p)\sin(2\theta_p) , \qq 
M(p) = 4ih\cos(2\varLambda_p) \sin(2\theta_p) ,
$$
where $\varLambda_p$ is what we call the pseudorapidity, and $\theta_p (= p_{\rm w.s.}/4)$ is the worldsheet momentum of an individual magnon. Since $R(\theta_1,\theta_2,\varLambda)=P$ only if $\theta_1=\theta_2$ and $\varLambda(=\varLambda_1-\varLambda_2)=0$ (in the principal region), the transmission channel is only allowed for magnons having the same pseudorapidity. The modes having no angular momentum, $M(p)=0$, (i.e.~massless modes in backgrounds with no flux) are obtained by setting the pseudorapidity to $\varLambda_p = \pi/4$ or $\varLambda_p = 3\pi/4$ (in the principal region). Their scattering is described by $R(\theta_1,\theta_2,\pi/2)$ and $R(\theta_1,\theta_2,0)$ (this, for example, exactly reproduces the two branches of scattering of massless modes). 

There are several other possible directions of further study. First, it would be interesting to generalize the constructions presented in this paper for centrally extended superalgebras $\mfsl({1|1})^n\op\C^n$, with $n>2$ for different ``linkings'' by extending \eqref{a:Lie}, namely $[e_i,f_j] = \delta_{ij} h_i + a_{ij} (1-\delta_{ij}) k_i$, where $(a_{ij})_{1\le i,j\le n}$ is the matrix of a connected graph and $[\cdot\hspace{.5mm},\cdot]$ denotes the graded commutator; the question we want to ask is what types of graphs lead to ``interesting'' $u$-deformed Hopf algebras having a $u$-deformed coproduct \eqref{a:cop}. Then we would like to compare the representations of these new algebras with the classification obtained in \cite{Kac}. Second, it would be interesting to construct a Drinfeld New presentation of $ U_q(\widehat\mfa)$ following the construction presented in \cite{He}. Third, a similar superalgebra $\mfsl({1|1})\op{\C u^{\pm}}$ emerges in the $AdS_2 \times S^2$ duality \cite{HPT,GH}. We hope the present paper will serve as a guideline for analogous constructions in this duality. Moreover, the $u$-deformed algebras are known to have additional so-called ``secret symmetries'' \cite{MMT,LMMRT,PTW}. The role of the ``secret symmetry'' of the Yangian $\mcY(\mfa)$ is played by the generating function $h_0(u)$; it would be interesting to find its analogue for the quantum affine algebra $ U_q(\widehat\mfa)$. Following the arguments presented in \cite{LRT}, it is natural to expect that there exist two ``secret symmetry'' extensions of $ U_q(\widehat\mfa)$ that in the evaluation representation are symmetries of (i.e.~intertwine with) the deformed $R$-matrix. Lastly, it is well known that both $AdS_5/CFT_4$ spin chain and its $q$-deformed model are closely related to the one-dimensional Hubbard model \cite{Be2,MaMe,RG} and its deformation \cite{BK,BGM}; we believe that the $AdS_3/CFT_2$ spin chain and its deformation can also be linked to the one-dimensional Hubbard model or some generalization thereof ({\it e.g.\@} \cite{FFR}). For example, one could consider the quotient of \eqref{a:Lie} by the ideal $<e_i-a_i,f_i-a^\dag_i,h_i-1,k_i-\al_i>$, with $\al_i\in\C$, giving the algebra $[a_i,a_j]=0$, $[a^\dag_i,a^\dag_j]=0$, $[a_i,a_j^\dag]=\del_{ij} + (1-\del_{ij})\al_i$, which can be interpreted as an ``algebra of interacting electrons''. \smallskip

{\it Acknowledgements.} The author thanks Andrea Prinsloo, Eric Ragoucy and Alessandro Torrielli for many useful suggestions and comments. Part of this work was done during the author’s visit to the Centro di Ricerca Matematica (CRM) Ennio De Giorgi for the research program ``Perspectives in Lie Theory''. The author thanks the CRM Ennio De Giorgi for the hospitality and support, and also gratefully acknowledges the UK EPSRC for the Postdoctoral Fellowship under the grant EP/K031805/1 ``New Algebraic Structures Inspired by Gauge/Gravity Dualities''.


\appendix


\section{Computing $R$-matrix}

\subsection{R-matrix} \label{App:R}

Let $R\in\End(V\ot V')$ be an arbitrary $4\times4$ matrix with elements $r_{ij}$ and $1\le i,j\le4$. We need to solve the intertwining equation
\eq{ \label{intw}
(\pi \ot \pi')(\Delta^{\rm op}(a))\, R = R\, (\pi \ot \pi')(\Delta(a)) 
}
for all $a\in U(\mfa)$.	Since the tensor product of two atypical modules is isomorphic to the typical module we can restrict the matrix $R$ to
\eq{ \label{R0}
R = r_{11} E_{11}\ot E_{11} + r_{22} E_{11}\ot E_{22} + r_{23} E_{12}\ot E_{21} + r_{32} E_{21}\ot E_{12} + r_{33} E_{22}\ot E_{11} + r_{44} E_{22}\ot E_{22}.    
}
Let $a=e_1$. Then \eqref{intw} is equivalent to the following set of linear equations:
\eqn{
\ga' \nu'(r_{11}-\nu ^2 r_{22})+\ga \nu r_{32} &= 0, \qq
\ga\nu({\nu}^{\prime2} r_{22} - r_{44}) - \ga' \nu' r_{32} = 0, \\
\ga (\nu^{\prime2} r_{11} - r_{33}) - \ga' \nu \nu' r_{32} &= 0 , \qq
\ga' (r_{33}- \nu^2 r_{44}) + \ga\nu \nu' r_{23} = 0 ,
}
having a solution
$r_{23} = {\ga'\nu'(\ga\nu)^{-1}(\nu^2 r_{22}-r_{11})}$, $r_{33} = \nu'(\nu'r_{11} - \ga^{-1}\ga'\nu r_{32})$, $r_{44} = \nu^{\prime2} r_{22} - \ga'\nu' (\ga\nu)^{-1} r_{32}$. 
Next, let $a=e_2$. Now \eqref{intw} gives
\eqn{
 (\ga^{\prime2} \nu^{\prime2}-\ga^2 \nu^2)r_{11} +  (\ga^2 -\nu^2 \ga^{\prime2} \nu^{\prime2}) r_{22} &= 0, \\
\nu'(\ga^2 \nu^2-\ga^{\prime2} \nu^{\prime2}) r_{32} + \ga \ga'\nu (\nu^{\prime4}-1) r_{22} &= 0, \\
\nu'(\ga^2 - \nu^2 \ga^{\prime2} \nu^{\prime2}) r_{32} + \ga \ga'\nu (\nu^{\prime4}-1) r_{11} &=0, \\ 
\ga^2 \nu \nu^{\prime2} (r_{22}-\nu^2 r_{11})+ (\nu^4-1) \ga\ga'\nu'r_{32}+\ga^{\prime2}(\nu r_{11} -\nu^3r_{22}) &= 0,
}
the solution of which is 
$$
r_{22} = (\ga^2\nu^2-\ga^{\prime2} \nu^{\prime2})(\ga^2 - \ga^{\prime2}\nu^2 \nu^{\prime2})^{-1}r_{11},\qu 
r_{32} =(\ga\ga'\nu (\nu^{\prime4} - 1))(\ga^{\prime2} \nu^2 \nu^{\prime3}-\ga^2\nu')^{-1}r_{11}.
$$
Finally, upon setting $r_{11}=\ga'\ga^{-1}\nu\nu'-\ga(\ga'\nu\nu')^{-1}$, we obtain \eqref{R}. It remains to check that \eqref{intw} holds when $a=f_i$, which follows by similar computations. 
%


\subsection{$q$-deformed R-matrix} \label{App:Rq}

This time we need to solve the intertwining equation \eqref{intw} for all $a\in U_q(\mfa)$. We restrict the matrix $R$ to the form given in \eqref{R0} and choose $a=E_1$. Then \eqref{intw} gives
\eqn{
q^{-\frac{\la_1}2} \ga'\nu'(q^{\la_1} \nu^2 r_{22} - r_{11}) - q^{-\frac{{\la}'_1}2} \ga \nu r_{23} &= 0, \qq
\ga( q^{{\la}'_1} \nu^{\prime2} r_{11} - r_{33}) - q^{\frac{\la_1+\la'_1}2} \ga'\nu \nu' r_{32} = 0, \\
q^{-\frac{\la_1}2} \ga'\nu' r_{32} + q^{-\frac{\la'_1}2} \ga \nu r_{44} - q^{\frac{\la'_1}2}  \ga \nu \nu^{\prime2} r_{22} &= 0 , \qq
q^{\frac{{\la}_1+\la'_1}2} \ga \nu \nu'r_{23} + \ga'(r_{33} - q^{\la_1} \nu^2 r_{44} ) = 0 ,
}
having a solution 
\begin{gather*}
r_{23} = q^{-\frac{\la_1-\la'_1}2} \, \ga'\nu'(\ga \nu)^{-1} \left( q^{\la_1} \nu^2 r_{22} - r_{11}\right), \qu
r_{33} = q^{\la'_1} \nu^{\prime2} r_{11} - q^{\frac{\la_1+\la'_1}2} \ga'\ga^{-1} \nu \nu' r_{32}, \\
r_{44} = q^{\la'_1} \nu^{\prime2} r_{22} - q^{\frac{\la'_1-\la_1}2}\, \ga'\nu'(\ga\nu)^{-1} r_{32}.
\end{gather*}
Now let $a=E_2$. Then \eqref{intw} for the remaining elements gives
\eqn{
 \left( q^{\frac{\la'_1-\la'_2-\la_1}2} \ga^{\prime2} \ga^{-2} \nu^{\prime2} - q^{-\frac{\la_2}{2}} \nu^2 \right) r_{11} + \left(q^{\frac{\la_2}{2}} - q^{\frac{\la'_1-\la'_2+\la_1}2  }\ga^{\prime2} \ga^{-2} \nu^2 \nu^{\prime2} \right) r_{22} &= 0, \\[0.25em]
q^{\frac{\la'_2+\la_2}2} \ga^2 \nu r_{32} - q^{\frac{\la'_1+\la_1}2} \ga^{\prime2}  \nu^2 \nu^{\prime3} r_{32} - \ga \ga'\nu \big(q^{\la'_2}-\nu^{\prime4} q^{\la'_1}\big) r_{11}  &= 0, \\[0.25em]
 \left( q^{-\frac{\la_2}{2}} \ga^{\prime-1} \nu  -  \ga^{-2} \ga'\nu^{-1} \nu^{\prime2} q^{\frac{1}{2} \left(\la'_1-\la'_2-\la_1\right)}\right) r_{32} -  q^{-\frac{\la'_2}{2}} \ga^{-1} {\nu }^{\prime-1} \big(q^{\la'_2}-\nu^{\prime4} q^{\la'_1}\big) r_{22} &= 0, \\[0.25em]
\ga \ga'\nu'\left(q^{\la_2}-\nu ^4 q^{\la_1}\right) r_{32} + q^{\frac{\la'_1+\la_1}2} \ga^2 \nu \nu^{\prime2}  \left(\nu^2 r_{11}- q^{\la_2} r_{22} \right) + q^{\frac{\la'_2+\la_2}2} \ga^{\prime2}  \nu \left(q^{\la_1} \nu^2 r_{22} - r_{11}\right) &=0 ,
}
the solution of which is
$$
r_{22} = \frac{ \ga^2 \nu^2 q^{\frac{\la'_2+\la_1}2}-\ga^{\prime2} \nu^{\prime2} q^{\frac{\la'_1+\la_2}2}}{\ga^2 q^{\frac{\la'_2+\la_1+2 \la_2}2}-\nu^2 \ga^{\prime2} \nu^{\prime2} q^{\frac{\la'_1+2 \la_1+\la_2}2}} \, r_{11}, \qq
r_{32} = \frac{\ga \nu \ga'\big(q^{\la'_2}-\nu^{\prime4} q^{\la'_1}\big)}{\ga^2 \nu' q^{\frac{\la'_2+\la_2}2} - \ga^{\prime2} \nu^2 \nu^{\prime3} q^{\frac{\la'_1+\la_1}2}} \, r_{11}  .
$$
Then, upon setting $r_{11}=q^{\frac{\la_1-\la'_2}2} \ga'\ga^{-1}\nu\nu' - q^{\frac{\la_2-\la'_1}2} {\ga}({\ga'\nu\nu'})^{-1}$, we obtain \eqref{Rq}. It remains to check that \eqref{intw} holds when $a=F_i$, which follows by a lengthy but direct computations and the usage of the identities \eqref{mumu} and \eqref{gaga} for parameters $\ga$ and $\ga'$.


\section{Traditional notation} \label{App:AdS}

We briefly recall the traditional notation used to describe the atypical module of $U(\mfa)$ in the literature on $AdS_3/CFT_2$ duality. We will mostly refer to \cite{BSS1,BSSS2}, where atypical modules are conveniently called short representations. Depending on the choice of the grading of vectors $w_0$ and $w_1$ ({\it c.f.} \eqref{atypical-mod}) these will be called left-moving or right-moving modules. We will use barred notation to describe the right-moving module, i.e.~$\bar w_0$, $\bar w_1$, etc.~(as in \cite{BSS1}; a tilde notation and subsripts {\tiny L} and {\tiny R} are used in \cite{BSSS2} instead). Note that generators of $U(\mfa)$ in the traditional notation can be identified with our notation by (see Appendix B: {\it most symmetric frame} in \cite{BSSST1})
$$
e_1 = \mf{Q}_{L} , \qu\! f_1 = \mf{S}_{L}, \qu\! e_2 = \mf{S}_{R}, \qu\! f_2 = \mf{Q}_{R}, \qu\! h_1 = \mf{H}_{L} , \qu\! h_2 = \mf{H}_{R} , \qu\! k_1 = \mf{P} , \qu\! k_2 = \mf{P}^\dagger , \qu\! u^{\pm} = e^{\pm i\frac{p}{4}}.
$$

\subsection{Left-moving module} \label{sec:left}

Consider the atypical module $A(\la_1,\la_2,\nu)={\rm span}_\C\{w_1,w_0\}$ defined in \eqref{atypical-mod} and choose $\deg_2 w_1=0$ and $\deg_2 w_0=1$ (this corresponds to the same setup as in Section \ref{Sec:2.6}). We introduce the notation 
\eq{
|\phi_p\rangle = \ga d_p w_1, \qq |\psi_p\rangle = w_0 , \label{left-basis}
}
for some $d_p \in \C^\times$, where $p=p_{\rm w.s.}$ denotes the worldsheet momentum. Vectors $|\phi_p\ran$ and $|\psi_p\ran$ are interpreted as bosonic and fermionic left-moving worldsheet magnons, respectively (see e.g.~\cite[Sec.~3.2]{BSS1}). Set $a_p = \ga^2 d_p$, $b_p = \mu_2/d_p$ and $c_p = \ga^{-2}\mu_1/d_p$. It follows from \eqref{mod} that 
\eqa{ \label{lmod}
e_1 \,|\phi_p\rangle = a_p \,|\psi_p\rangle , \qu 
f_1 \,|\psi_p\rangle = b_p \,|\phi_p\rangle , \qu
f_2 \,|\psi_p\rangle = c_p \,|\phi_p\rangle , \qu
e_2 \,|\phi_p\rangle = d_p \,|\psi_p\rangle ,
}
and we recover the parametrization used in \cite[Sec.~4.1]{BSS1}. Recall that $\la_1= \ga^2\mu_2 $, $\la_2= \ga^{-2}\mu_1$ and  $\mu_i = \al_i(\nu^2 - \nu^{-2})$ for $i\in\{1,2\}$. By requiring the module to be unitary we find $a_p^*= b_p$, $c_p^*=d_p$, $\nu^*=\nu^{-1}$ and $\mu_i^* = \mu_j$ for $i,j\in\{1,2\}$ and $i\ne j$. Fix $p,M \in \C$, where $M$ denotes the angular momentum of a magnon (see e.g.~\cite[Sec.~2.4.2 and Sec.~4.2]{BSSS2} for the description of $M$ for massive and massless modes) and introduce parameters $x^\pm_p \in \C^\times$ (usually called Zhukovski variables, see \cite{Be1}) satisfying
\eq{
\frac{x^+_p}{x^-_p} = e^{\sqrt{-1}\,p} , \qq
x^+_p + \frac{1}{x^+_p} - x^-_p - \frac{1}{x^-_p} =  \frac{\sqrt{-1}\,M}{h} , \label{mass-shell}
}
where $h\in\C$ is identified with the coupling constant of the model. 
Then, without loss of generality, we can choose $\al_1=-\al_2=-h$ and $\nu^4 = \frac{x^+_p}{x^-_p}$ such that $h^*=h$ and $(x^\pm_p)^*=x^\mp_p$. This gives 
\eq{ \label{param1}
\mu_1 \,(=a_p c_p)= h \nu^2\left(\frac{x^-_p}{x^+_p}-1\right)  , \qq \mu_2 \,(=b_p d_p)= h \nu^{-2}\left(\frac{x^+_p}{x^-_p}-1\right) .
}
Set $\eta_p^2=\sqrt{-1}\,(x^-_p - x^+_p)$. Then the parametrization
\eq{ \label{abcd}
a_p = \sqrt{h}\, \eta_p \,\nu_p , \qq 
b_p = \sqrt{h}\, \frac{\eta_p}{\nu_p}, \qq 
c_p = - \sqrt{-h}\, \frac{\eta_p \nu_p}{x^+_p} , \qq 
d_p = \sqrt{-h}\, \frac{\eta_p }{x^-_p \nu_p} .  
}
satisfies the required constraints, since $\eta_p^* = \eta_p$. (Here we have denoted $\nu_p=\nu$ for homogeneity of the notation.) Thus we find that
\eq{ \label{param2}
\ga^2(= a_p/d_p) = - \sqrt{-1} \, \nu^2_p x^-_p , \;\;
\la_1 (=a_p b_p)= \sqrt{-1} h\, (x^-_p - x^+_p ) , \;\; 
\la_2 (=c_p d_p) = \sqrt{-1} h \left( \frac{1}{x^+_p} - \frac{1}{x^-_p}\right) \!,\!
}
which can be used to rewrite the $R$-matrix \eqref{R} (or equivalently \eqref{R:trig}) in the traditional notation. Note that constraint $\la_1\la_2 = \mu_1 \mu_2$ becomes trivial in this parametrization.


\subsection{Right-moving module} \label{sec:right}

Consider the atypical module $A(\bar\la_1,\bar\la_2,\bar\nu)={\rm span}_\C\{\bar w_1,\bar w_0\}$, such that $\deg_2 \bar w_0=0$ and $\deg_2 \bar w_1=1$. Introduce the notation 
\eq{
|\bar\psi_p\rangle = \bar\ga\bar b_p \bar w_1, \qq |\bar\phi_p\rangle = \bar w_0 ,  \label{right-basis}
}
for some $\bar b_p \in \C^\times$. Vectors $|\bar \phi_p\ran$ and $|\bar \psi_p\ran$ can be interpreted as bosonic and fermionic right-moving worldsheet magnons, respectively. Similarly as before, we introduce parameters $\bar x^\pm_p$ and angular momentum $\bar M$ satisfying analogous relations to those in \eqref{mass-shell}. Then, by requiring an analogous relation to \eqref{param1} to hold, namely $[e_1,f_2]\,|\bar\varphi_p\rangle = \bar a_p \bar c_p\,|\bar\varphi_p\rangle$ and $[e_2,f_1]\,|\bar\varphi_p\rangle= \bar b_p \bar d_p\,|\tl\varphi_p\rangle$ for $\bar\varphi\in\{\bar\phi,\bar\psi\}$, we find
\spl{ \label{rmod}
e_1 \,|\bar\psi_p\rangle = \bar c_p \,|\bar\phi_p\rangle , \qu 
f_1 \,|\bar\phi_p\rangle = \bar d_p \,|\bar\psi_p\rangle , \qu
f_2 \,|\bar\phi_p\rangle = \bar a_p \,|\bar\psi_p\rangle , \qu
e_2 \,|\bar\psi_p\rangle = \bar b_p \,|\bar\phi_p\rangle ,
}
where parameters $\bar a_p$, $\bar b_p$, $\bar c_p$, $\bar d_p$ have the same explicit form as those in \eqref{abcd} except $x^\pm_p$ are substituted by $\bar x^\pm_p$, and similarly for $\eta_p$ and $\nu_p$ (as in \cite[Sec.~4.2]{BSS1}). Next, by comparing the expression above with \eqref{mod}, we find that $\bar c_p = \bar\ga^2\,\bar b_p$, $\bar d_p = \bar\mu_2 / \bar b_p$, $\bar a_p = \bar\ga^{-2}\bar \mu_1 / \bar b_p$ and thus 
\eq{ \label{param3}
\bar\ga^2 (= \bar c_p/\bar b_p) = -\sqrt{-1} \frac{\bar\nu^2_p}{\bar x^+_p}, \qu 
\bar \la_1 (=\bar  c_p \bar d_p)= \sqrt{-1}h \left( \frac{1}{\bar x^+_p} - \frac{1}{\bar x^-_p}\right) , \qu 
\bar\la_2 (=\bar a_p \bar b_p)= \sqrt{-1}h (\bar x^-_p - \bar x^+_p ) .
}
which together with \eqref{param2} can be used to write $R$-matrices given by \eqref{R:trig-2} in the traditional notation. Also, as before, the constraint $\bar\la_1\bar\la_2 = \bar\mu_1 \bar\mu_2$ is trivial in this parametrization.


\section{Traditional deformed notation} \label{App:q-AdS}

Here we present the traditional notation (as in \cite{BGM,Ho}) that can be used to describe the atypical module of the $q$-deformed algebra $ U_q(\mfa)$. Generators of $ U_q(\mfa)$ in the traditional notation (as in \cite{Ho}) can be identified with our notation by $U^{\pm2} = \mf{U}^{\pm1}$ and
\eqn{
& E_1 = (\mf{U}\,\mf{V}_{\sm{$L$}})^{-\frac{1}{2}}\mf{Q}_+, \qu F_1 = (\mf{U}\,\mf{V}_{\sm{$L$}})^{\frac{1}{2}}\mf{S}_-, \qu K^{\pm}_1 = \mf{V}^{\pm\frac{1}{2}}_{\sm{$L$}}, \qu \al_1(L^+_1 - L^-_1) = (q-q^{-1})\,\mf{U}^{-1}(\mf{V}_{\sm{$L$}}\mf{V}_{\sm{$R$}})^{-\frac{1}{2}}\mf{P},\\
& E_2 = (\mf{U}\,\mf{V}_{\sm{$R$}})^{\frac{1}{2}}\mf{S}_+, \qu F_2 = (\mf{U}\,\mf{V}_{\sm{$R$}})^{-\frac{1}{2}}\mf{Q}_-, \qu  K^{\pm}_2 = \mf{V}^{\pm\frac{1}{2}}_{\sm{$R$}}, \qu \al_2(L^+_2 - L^-_2) = (q-q^{-1})\,\mf{U}^{+1}(\mf{V}_{\sm{$L$}}\mf{V}_{\sm{$R$}})^{\frac{1}{2}}\mf{K}.
}

\subsection{Left-moving module} \label{sec:qleft}

Consider the atypical module $A_q(\la_1,\la_2,\nu)={\rm span}_\C\{w_0,w_1\}$ and choose $\deg_2 w_1=0$ and $\deg_2 w_0=1$. Following the steps in Section \ref{sec:left} we define new vectors
$|\phi_p\rangle= \ga d_p w_1$ and $|\psi_p\rangle= w_0$ for some $d_p\in\C^\times$. Set $a_p = \ga^2 d_p$, $b_p = \al_2[\mu_2]_q/d_p$ and $c_p = \al_1 \ga^{-2}[\mu_1]_q/d_p$. It follows from \eqref{qmod} that 
\spl{ \label{lqmod}
E_1 \,|\phi_p\rangle = a_p \,|\psi_p\rangle , \qu 
F_1 \,|\psi_p\rangle = b_p \,|\phi_p\rangle , \qu
F_2 \,|\psi_p\rangle = c_p \,|\phi_p\rangle , \qu
E_2 \,|\phi_p\rangle = d_p \,|\psi_p\rangle .
}
Denote $\sigma = q^{(\la_1+\la_2)/4}$ and $\del = \la_1 - \la_2$, and set $\al_1=\al_2=h$. We have ($i,j\in\{1,2\}$, $i\ne j$)
\eq{ \label{qmod-cons}
[E_i, F_i]\,|\varphi_p\rangle = [\la_i]_q\, |\varphi_p\rangle, \qu\; 
[E_i, F_{j}]\,|\varphi_p\rangle = h\, [\mu_i]_q\, |\varphi_p\rangle 
\;\tx{for}\; \varphi_p \in \{\phi_p,\psi_p\}.
}
Thus the representation labels must satisfy the following set of identities
\spl{ \label{q:cons}
a_p b_p \,(= [\la_1]_q) &= \frac{q^{\del/2}\si^2 - q^{-\del/2}\si^{-2}}{q-q^{-1}} , \qq a_p c_p \,(= \al_1[\mu_1]_q) = h\, \frac{\nu^2\si^2 - \nu^{-2}\si^{-2}}{q-q^{-1}} ,\\
c_p d_p \,(= [\la_2]_q) &= \frac{q^{-\del/2}\si^2 - q^{\del/2}\si^{-2}}{q-q^{-1}} , \qq b_p d_p \,(= \al_2[\mu_2]_q) = h\, \frac{\nu^{-2}\si^{2} - \nu^{2}\si^{-2}}{q-q^{-1}} . 
}
Moreover, the module shortening constraint given in \eqref{gaga} becomes
\eq{ \label{q:short}
h^2 \left(\nu^2\si^2 - \nu^{-2}\si^{-2} \right)\left( \nu^{-2}\si^{2} - \nu^{2}\si^{-2}\right) = \left( \si^2 - q^{-\del}\si^{-2} \right)\left( \si^2 - q^{\del}\si^{-2} \right).
}
Inspired by \cite{BGM} we choose the following $x^\pm$-parametrization:
\spl{ \label{q:vs}
\nu^4 = q^\del \frac{x^+}{x^-}\frac{\xi x^- + 1}{\xi x^+ + 1} = q^{-\del} \frac{x^+ + \xi}{x^- + \xi} , \qq \si^4 = q^\del \frac{x^+}{x^-}\frac{x^- + \xi }{x^+ + \xi} = q^{-\del} \frac{\xi x^+ + 1}{\xi x^- + 1},	
}
where parameters $x^\pm$ and $\xi$ satisfy
\spl{ \label{q:zeta}
q^{-\del}\zeta(x^+) = q^\del \zeta(x^-), \qq  \zeta(x) = -\frac{x+x^{-1}+\xi + \xi^{-1}}{\xi - \xi^{-1}} , \qq h^2 = \frac{\xi^2}{\xi^2-1}.
}
It is a direct computation to verify that this parametrization satisfies \eqref{q:short}. Using an analogy to \eqref{abcd}, we set $\eta_p^2=i(x^-_p-x^+_p)$ and fix the expression for $a_p$ to give
\spl{ \label{q:abcd}
a_p &= \sqrt{h}\, \eta_p \nu_p \si_p , \qq
b_p = i\sqrt{h}\, \frac{q^{\del/2} \xi\, \eta_p \si_p }{h \nu_p (q-q^{-1}) (\xi x^+_p+1)} , \\
c_p &= i\sqrt{h}\, \frac{\eta_p\nu_p\si_p}{(q-q^{-1})\,x^+_p}, \qq
d_p = \sqrt{h}\, \frac{q^{-\del/2} \xi\, \eta_p \si_p}{h \nu_p (x^-_p+\xi)} ,
}
where we have added extra subscripts to all the parameters for the homogeneity of the notation. We have that $\ga^2 = a_p / d_p$, which together with the relations (\ref{q:cons}-\ref{q:abcd}) can be used to write $R$-matrix \eqref{Rq} in the traditional notation. 


\subsection{Right-moving module} \label{sec:qright}

Consider the atypical module $A_q(\bar\la_1,\bar\la_2,\bar\nu) = {\rm span}_\C\{\bar w_0,\bar w_1\}$ and set $\deg_2 \bar w_0=0$ and $\deg_2 \bar w_1=1$. As in Section \ref{sec:right}, we introduce the notation $|\bar\psi_p\ran= \bar\ga\bar b_p \bar w_{1}$ and $|\bar\phi_p\rangle= \bar w_0$. Then, requiring ($i,j\in\{1,2\}$, $i\ne j$)
\eq{ 
[E_i, F_i]\,|\bar\varphi_p\ran = [\bar\la_i]_q\, |\bar\varphi_p\ran, \qu\; 
[E_i, F_{j}]\,|\varphi_p\ran = h\, [\bar\mu_i]_q\, |\bar\varphi_p\ran 
\;\tx{for}\; \bar\varphi \in \{\bar\phi,\bar\psi\}. \label{}
} 
and $\bar a_p \bar c_p = h\,[\bar \mu_1]_q$ and $\bar b_p \bar d_p = h\,[\bar \mu_2]_q$ ({\it c.f.\@} \eqref{q:cons}) we find
\spl{ \label{rqmod}
E_1 \,|\bar\psi_p\ran &= \bar c_p \,|\bar\phi_p\rangle , \qu E_2 \,|\bar\psi_p\ran = \bar b_p \,|\bar\phi_p\rangle , \qu F_1 \,|\bar\phi_p\ran = \bar d_p\,|\bar\psi_p\rangle , \qu 
F_2 \,|\bar\phi_p\ran = \bar a_p\,|\bar\psi_p\rangle , 
}
where parameters $\bar a_p$, $\bar b_p$, $\bar c_p$, $\bar d_p$ have the same explicit form as those in \eqref{q:abcd} subject to the bar notation. We have that $\bar\ga^2 = \bar c_p / \bar b_p$, which together with (\ref{q:cons}-\ref{q:abcd}), subject to the bar notation, can be used to write deformed analogues of the $R$-matrices discussed Remark \ref{R:R} (3) in the traditional notation.


\section{Coproduct $\Delta_\eps$} \label{App:Cop}

The coproduct $\Delta_\eps$ is given by ($i,j\in\{1,2\}$, $i\ne j$)
\spl{ \label{Delta:e}
\Delta_\eps(e_{i,r}) &= e_{i,r} \ot u^{\dt\mp} + u^{\dt\pm} \ot e_{i,r} + \sum_{l=1}^r \left( \eps_i\, u^{\dt\pm} h_{i,r-l}\ot e_{i,l-1} + \eps_{j}\, u^{\dt\mp} k_{i,r-l}\ot u^{\dt\mp2} e_{j,l-1} \right) , \\
\Delta_\eps(f_{i,r}) &= f_{i,r} \ot u^{\dt\pm} + u^{\dt\mp} \ot f_{i,r} + \sum_{l=1}^r \left( \eps_i f_{i,r-l} \ot u^{\dt\pm} h_{i,l-1} + \eps_{j}\, u^{\dt\mp2} f_{j,r-l} \ot u^{\dt\mp} k_{j,l-1} \right) , \\
\Delta_\eps(h_{i,r}) &= h_{i,r} \ot 1 + 1 \ot h_{i,r} + \sum_{l=1}^r \left( \eps_i h_{i,r-l}\ot h_{i,l-1} + \eps_{j}\, u^{\dt\mp2} k_{i,r-l}\ot u^{\dt\mp2} k_{j,l-1} \right) , \\
\Delta_\eps(k_{i,r}) &= k_{i,r} \ot u^{\dt\mp2} + u^{\dt\pm2} \ot k_{i,r} + \sum_{l=1}^r \left( \eps_{j}\, k_{i,r-l}\ot u^{\dt\mp2} h_{j,l-1} + \eps_i\,  u^{\dt\pm2} h_{i,r-l}\ot k_{i,l-1} \right) ,\\
\Delta(h_{0,r}) &= h_{0,r} \ot 1 + 1 \ot h_{0,r} - \sum_{l=1}^r \left( \eps_1\, u^{+} f_{1,r-l} \ot u^{+} e_{1,l-1} + \eps_2\, u^{-} f_{2,r-l} \ot u^{-} e_{2,l-1} \right) 
}
Setting $\eps_i=1$ one obtains $\Delta$ given by \eqref{Y(a):cop}.


\medskip

\end{document}